\documentclass[11pt]{article}
\usepackage{amsmath}
\usepackage{setspace}
\usepackage{graphicx}
\usepackage{caption}
\usepackage{subcaption}
\usepackage{amsmath}
\usepackage{amsthm}
\usepackage{enumerate}
\usepackage{amsfonts}
\usepackage{natbib}
\graphicspath{{Plots/}}
\newcommand{\disc}{\boldsymbol\eta}

\newcommand{\x}{\boldsymbol\theta}
\newcommand{\control}{\textbf{x}}
\newcommand{\f}{\textbf{f}}
\newcommand{\err}{\textbf{e}}

\newcommand{\obs}{\textbf{z}}

\newcommand{\sys}{\textbf{y}}

\newcommand{\best}{\boldsymbol\theta^{*}}
\newcommand{\ens}{\textbf{F}}
\newcommand{\ensc}{\textbf{F}_{\boldsymbol{\mu}}}
\newcommand{\bas}{\boldsymbol\Gamma}
\newcommand{\bass}{\boldsymbol\gamma}
\newcommand{\resid}{\boldsymbol\epsilon}
\newcommand{\rot}{\boldsymbol\Lambda}
\newcommand{\rott}{\boldsymbol\lambda}
\newcommand{\var}{\boldsymbol\Sigma}
\newcommand{\weight}{\textbf{W}}
\newcommand{\B}{\textbf{B}}
\newcommand{\rr}{\textbf{r}}
\newcommand{\cc}{\textbf{c}}
\newcommand{\bb}{\textbf{b}}
\newcommand{\R}{\mathcal{R}}
\newcommand{\V}{\mathcal{V}}
\newcommand{\ibound}{T}
\newcommand{\bassf}{\boldsymbol\varphi}
\newcommand{\resens}{\textbf{F}_{\epsilon}}
\newcommand{\resbas}{\B_{\epsilon}}
\newcommand{\tspace}{\Theta}
\newtheorem{result}{Result}
\newtheorem{theorem}{Theorem}

\begin{document}

\title{\bf Uncertainty quantification for computer models with spatial output using calibration-optimal bases}
	\author{James M. Salter\thanks{
			The authors gratefully acknowledge \textit{support from EPSRC fellowship No. EP/K019112/1 and support from the NSERC funded Canadian Network for Regional Climate and Weather Processes (CNRCWP). We would also like to thank Yanjun Jiao for managing our ensembles of CanAM4.}}\hspace{.2cm}\\
			Department of Mathematics, College of Engineering, Mathematics and \\ Physical Sciences, University of Exeter, UK.\\
			and \\
			Daniel B. Williamson \\
			Department of Mathematics, College of Engineering, Mathematics and \\ Physical Sciences, University of Exeter, UK.\\
			and \\
			John Scinocca \\
			Canadian Centre for Climate Modelling and Analysis \\ Victoria, Canada.\\
			and\\
			Viatcheslav Kharin \\
			Canadian Centre for Climate Modelling and Analysis \\ Victoria, Canada.
		}
\maketitle

\begin{abstract}
The calibration of complex computer codes using uncertainty quantification (UQ) methods is a rich area of statistical methodological development. When applying these techniques to simulators with spatial output, it is now standard to use principal component decomposition to reduce the dimensions of the outputs in order to allow Gaussian process emulators to predict the output for calibration. We introduce the `terminal case', in which the model cannot reproduce observations to within model discrepancy, and for which standard calibration methods in UQ fail to give sensible results. We show that even when there is no such issue with the model, the standard decomposition on the outputs can and usually does lead to a terminal case analysis. We present a simple test to allow a practitioner to establish whether their experiment will result in a terminal case analysis, and a methodology for defining calibration-optimal bases that avoid this whenever it is not inevitable. We present the optimal rotation algorithm for doing this, and demonstrate its efficacy for an idealised example for which the usual principal component methods fail. We apply these ideas to the CanAM4 model to demonstrate the terminal case issue arising for climate models. We discuss climate model tuning and the estimation of model discrepancy within this context, and show how the optimal rotation algorithm can be used in developing practical climate model tuning tools.
\end{abstract}

\section{Introduction}
\label{sec:intro}

The design and analysis of computer experiments, now part of a wider cross-disciplinary endeavour called `Uncertainty Quantification' or `UQ', has a rich history in statistical methodological development as far back as the landmark paper by \cite{sacks1989design}. The calibration of computer simulators, a term reserved for methods that locate simulator input values with outputs that are consistent with physical observations (the inverse problem), is a well studied problem in statistical science, with Kennedy and O'Hagan's Bayesian approach based on Gaussian processes the most widely used \citep{kennedy2001bayesian}. 

The essence of the statistical approach to calibration is to combine a formal statistical model relating the computer simulator to real-world processes for which we have partial observations \citep{kennedy2001bayesian, goldstein2009reified, williamson2013history}, with a statistical representation of the relationship between inputs and outputs of the simulator based, typically, on Gaussian processes \citep{haylock1996inference}. 

Extensions for computer simulators with spatio-temporal output have centred around projecting the output onto a basis and adapting calibration methods to the lower-dimensional projections of these fields. Though wavelets \citep{bayarri2007computer} and B-splines \citep{williamson2012fast} have been tried, the approach due to \cite{higdon2008computer}, based on the principal components of the simulator output, has become the default method. Statistical methodological developments in UQ have built on principal component methods (e.g. \citet{wilkinson2010bayesian, chang2014probabilistic, chang2016calibrating}), and they have seen wide application, particularly in the analysis of climate models \citep{sexton2011multivariate, chang2014probabilistic, pollard2016}.

What statisticians term calibration is referred to as `tuning' in the climate modelling community, a process that has a huge influence on the projections made by each modelling centre and by the Intergovernmental Panel on Climate Change \citep{ipcc}. Each modelling centre submits integrations of their climate model for 4 different forcing scenarios (known as Representative Concentration Pathways) to each phase of the Coupled Model Intercomparison Project \citep{meehl2000coupled}, with the input parameters of the model `tuned' prior to submission so that the model output compares favourably with certain key observations. The resulting integrations, and not the simulators themselves, are what most climate scientists call `climate models' (i.e. simulators are not considered to be functions of these now fixed parameters). These integrations are used to discover physical mechanisms \citep{scaifeetal12}, projected trends \citep{screenwilliamson17}, drivers of variability \citep{collinsetal10} and future uncertainty to aid policy making \citep{harrisetal06}. 

Despite the application of UQ methods to the calibration of `previous-generation' climate models, referred to in the papers above and many others, UQ is not used for tuning within any of the major climate modelling centres \citep{hourdinetal16}. Instead, climate model parameters are often explored individually and tuning done by hand and eye, with the parameters changed, and the new run either accepted or rejected based on heuristic comparison with the current `best' integration. Different descriptions of these processes are offered by \cite{mauritsen2012tuning, williamson2017tuning, hourdinetal16}. 

This lack of uptake of state-of-the-art statistical methodology for calibration amongst some of the world's most important computer simulators should give us pause for thought. The `off-the-shelf' methodology, Bayesian calibration with principal components, is widely used elsewhere, well published, and is applied to many lower resolution climate models within the climate science literature. Is the lack of uptake a communication issue, or are there features of our methodology that mean it doesn't scale up well to climate simulators?

In this paper we show how the terminal case, wherein a simulator cannot be satisfactorily calibrated, manifests in the inference of standard UQ methodologies. We then demonstrate that even when there is a good solution to the inverse problem, the use of standard basis representations of spatial output (e.g. principal components across the design) can and regularly do lead to the terminal case and incorrect inference. We develop a simple test to see whether an analysis will lead to the terminal case before performing the calibration and, when the terminal case is not guaranteed, provide a methodology for finding an optimal basis for calibration, via a basis rotation. The efficacy of our methodology is demonstrated through application to an idealised example, and its relevance to climate model tuning through application to the calibration of the atmosphere of the current Canadian climate model, CanAM4.

In Section 2, we review UQ methodologies for calibration and present the terminal case for scalar model output. Section 3 reviews the standard approach to handling spatial output and demonstrates the implications of the terminal case for these methods through an idealised example. Section \ref{sectionbasis} presents novel methods for finding optimal bases for calibration that overcome the terminal case issues and demonstrates the efficacy of calibrating with optimal bases for our example. In Section 5 we see that standard approaches always lead to terminal analyses in CanAM4, and show how our optimal basis methodology can be used in the process of climate model tuning. Section 6 contains discussion.

\section{Calibration methodologies and the terminal case}
We consider a computer simulator to be a vector-valued function $f(\x, \control)$, with input parameters $\x$ that we wish to estimate/constrain, and `control' or `forcing' parameters, $\control$, both of which can be altered to perform computer experiments. For example, $\control$ might represent future CO$_2$ concentrations in a climate model. $f(\cdot,\control)$ simulates a physical system $\sys(\control)$, and we have access to measurements or observations $\obs$, of part or all of $\sys$. The goal of calibration methods is to use $\obs$ to learn about $\x$. In what follows we remove the control parameters, $\control$, to simplify the notation, as they are not involved in calibration, but in subsequent prediction. 

The two statistical methodologies for calibration that we focus on here are Bayesian (or probabilistic) calibration \citep{kennedy2001bayesian, higdon2008computer}, and history matching with iterative refocussing \citep{craig1996bayes, vernon2010galaxy, williamson2017tuning}. Both begin with the same type of assumption, namely that there exists a best input setting, $\best$, so that
\begin{equation}\label{realitymodel}
\sys = f(\best) + \disc, \qquad \obs = \sys + \err
\end{equation}
for mean-zero independent observation errors, $\err$, and model discrepancy, $\disc$ (though history matching differs in only requiring uncorrelated terms in (\ref{realitymodel}) rather than independent terms). 

Both methods require an emulator, usually a Gaussian process representation of function $f(\x)$, trained using runs $\ens = (f(\x_{1}),...,f(\x_{n}))$ based on design $\textbf{X} = (\x_1,\ldots,\x_n)$. For scalar $f(\cdot)$, the general model is
\begin{equation}\label{emulator}
f(\x) | \boldsymbol{\beta}, \boldsymbol{\phi} \sim \mathrm{GP}\left(\boldsymbol{\beta}^Tg(\x), R(|\x-\x'|;\boldsymbol{\phi})\right),
\end{equation}
where $g(\x)$ is a vector of specified regressors, $\boldsymbol{\beta}$ their coefficients, and $R(|\x-\x'|;\boldsymbol{\phi})$ a weakly stationary covariance function with parameters $\boldsymbol{\phi}$. The model is completed by specifying a prior on the parameters, $\pi(\boldsymbol{\beta}, \boldsymbol{\phi})$, and posterior inference given $\ens$ follows naturally with
	\begin{displaymath}
	f(\x)|\ens, \boldsymbol{\beta}, \boldsymbol{\phi} \sim \mathrm{GP}(m^*(\x), R^*(\cdot,\cdot; \boldsymbol{\phi}))
	\end{displaymath}
	with
	\begin{displaymath}
	\begin{split}
	m^*(\x) = \boldsymbol{\beta}^Tg(\x) + \textbf{K}(\x)\textbf{V}^{-1}\left(\ens -  \boldsymbol{\beta}^Tg(\textbf{X})\right), &\quad \textbf{K}(\x) = R(\x, \textbf{X}; \boldsymbol{\phi}),  \\
R^*(\x, \x'; \boldsymbol{\phi}) = R(\x, \x'; \boldsymbol{\phi}) - \textbf{K}(\x)\textbf{V}^{-1}\textbf{K}(\x')^T, &\quad \textbf{V} = R(\textbf{X}, \textbf{X}; \boldsymbol{\phi}).
	\end{split}
	\end{displaymath}
	There are many variants on emulation, with some practitioners preferring no regressors \citep{chen2016analysis}, different types of correlation function (including no correlation) \citep{kaufman2011efficient, salter2016comparison}, and different priors, $\pi(\boldsymbol{\beta}, \boldsymbol{\phi})$, with some leading to partially analytic posterior inference \citep{haylock1996inference}. As history matching only requires posterior means and variances of the emulator, Bayes linear analogues are sometimes used \citep{vernon2010galaxy}. Generalisations to multivariate Gaussian processes are natural \citep{conti2010bayesian}, and we address the difficulty with high dimensional output from Section \ref{spatial} onwards.
\subsection{Probabilistic calibration}	
Though the underlying statistical model and the emulator are similar for both history matching and probabilistic calibration, the assumptions placed upon $\best$, and the resulting inference, are quite different. Probabilistic calibration places a prior on $\best$, $\pi(\best)$, and a Gaussian process prior for the discrepancy, $\disc \sim \mathrm{GP}(0,\var_{\disc})$, before deriving the posterior $\pi(\best, \disc | \ens, \obs)$, and marginalising for $\best$. The discussion of \cite{kennedy2001bayesian}, and the later paper by \cite{brynjarsdottir2014learning}, argue that lack of identifiability between $\best$ and $\disc$ mean that strong prior information on $\disc$ or $\best$ is essential for effective probabilistic calibration to be possible.

\subsection{History matching and iterative refocussing}
Note that, given a discrepancy variance, probabilistic calibration must still give a posterior $\pi(\best | \ens, \obs)$ that integrates to $1$, thus predetermining an analysis that will point to some region of parameter space $\tspace$ as being `most likely'. This can be undesirable in some application areas, as often the goal is to find out if the simulator \textit{can} get `close enough' to the observations, so that experiments predicting the future can be trusted. Climate model tuning is a good example of this, where part of the goal in tuning is to find out whether it is the choice of parameters, or the parameterisation itself, that is leading model bias \citep{mauritsen2012tuning, hourdinetal16}. 

The method of history matching and iterative refocussing allows the question of whether the model is fit for purpose to be answered as part of the calibration exercise, by altering the problem from one of looking for the best input directly, to one of trying to rule out regions of $\tspace$ that could not contain $\best$. A model unfit for purpose would have all of $\tspace$ ruled out. The method defines an implausibility measure, $\mathcal{I}(\x)$, with
\begin{equation} \label{mvimpl}
\mathcal{I}(\x) = (\obs - \text{E}[f(\x)])^{T}(\text{Var}(\obs - \text{E}[f(\x)]))^{-1} (\obs - \text{E}[f(\x)]),
\end{equation}
where the expectations and variances of $f(\x)$ are derived from the Gaussian process emulator description above, and are conditioned on the runs $\ens$. If $\mathcal{I}(\x)$ exceeds a threshold, $\ibound$, that value of $\x$ is considered implausible and ruled out, thus defining a membership function for a subspace $\tspace'$ of $\tspace$ that is Not Ruled Out Yet (NROY), with $\tspace' = \left\{\x\in\tspace : \mathcal{I}(\x)\leq T \right\}.$ The choice of $\ibound$ will be problem dependent, though typically, if $\obs$ is one-dimensional, Pukelsheim's three sigma rule \citep{pukelsheim1994three} is used to set $\ibound = 9$ \citep{craig1996bayes, williamson2015bias}. For $\ell$-dimensional $\obs$, \citet{vernon2010galaxy} define $\ibound = \chi^2_{\ell, 0.995}$, the $99.5$th percentile of the $\chi^2$-distribution with $\ell$ degrees of freedom, or a conservative $T$ can be derived through Chebysev's inequality.

A key principle behind history matching is its iterative nature. Following an initial set of runs, a `wave' of history matching is conducted, leading to a certain percentage of $\tspace$ being ruled out. A new wave can then be designed within NROY space, and the procedure repeated, refocussing the search for possible $\best$ \citep{williamson2017tuning}.

Discrepancy and observation error variances, $\var_{\disc}$ and $\var_{\err}$, are important in both probabilistic calibration and history matching. For the latter, equation (\ref{realitymodel}) leads to
\begin{displaymath}
\text{Var}(\obs - \text{E}[f(\x)]) = \text{Var}[f(\x)] + \var_{\disc} + \var_{\err}
\end{displaymath}
in equation \eqref{mvimpl}, whilst a Normal assumption on $\err$ in calibration means $\var_{\disc}$ and $\var_{\err}$ appear in the likelihood.

In this paper, we focus on optimal spatial calibration for both types of methodology, as the issues we shall identify in Section \ref{spatial} apply equally to both, though manifest in different ways, as we shall illustrate now with our discussion of the terminal case.
\subsection{The terminal case} \label{terminal}

Consider a computer simulator, $f(\x)$, a discrepancy variance assessment $\var_{\disc}$, and an observation error variance $\var_{\err}$, where both variance matrices are positive definite. We define the terminal case to occur when $\mathcal{I}(\x) > T$, for $T$ as above and for a perfect emulator, so that, in equation \eqref{mvimpl}, $\text{E}[f(\x)]=f(\x)$ and $\text{Var}[f(\x)]=0$ for all $\x$. So, from a history matching perspective, the terminal case occurs when the model is too far from the observations at every point in parameter space according to the model discrepancy. Hence, all of $\tspace$ is ruled out, and the modellers must reconsider their simulator, or their error tolerance.

% \remove[dw]{so that $\obs$, under equation} (\ref{realitymodel}), is not consistent with our simulator for an $\x \in \tspace$. Within the history matching framework, the interpretation of $\var_{\disc}$ is of a tolerance of the modellers to model error, so that this situation, the `terminal case', implies that without introducing an emulator, $\mathcal{I}(\x) > T$ for all $\x$.} 

Within a probabilistic calibration framework, the terminal case implies a prior-data conflict so that, in some sense, $\var_{\disc}$ has been `misspecified' or the expert is `wrong'. Lack of identifiability requires informative expert judgement for discrepancy \citep{brynjarsdottir2014learning}, yet the difficulty in providing such judgements for complex computer simulators \citep{goldstein2009reified} may mean that the terminal case would occur quite often in practice. It is therefore important to see how such prior-data conflict would manifest. 
\begin{figure}[h!]
	\centering
	\includegraphics[width = 1\textwidth, height=0.69\textheight]{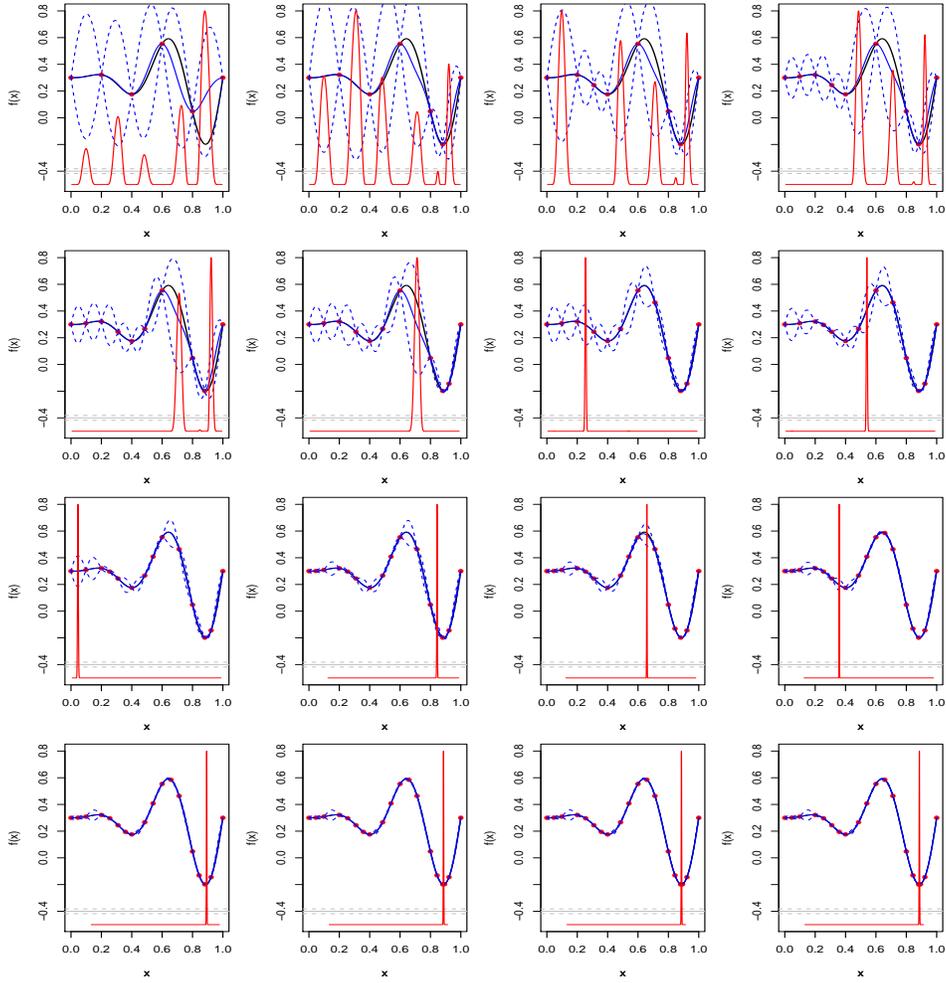}
	\caption{Showing 20 steps of an iterative probabilistic calibration of a computer simulator (black line) to observations (solid grey line) with $\var_{\disc}$ and $\var_{\err}$ misspecified (dashed grey lines $\pm 3$ standard deviations). Observations of the model (red dots) iteratively taken at the MAP estimate for $\best$ following the fitting of a GP emulator (mean solid blue line, $\pm 2$ standard deviations dashed blue lines), and the posterior distribution of $\best$ overlaid at each step (solid red line).}\label{whackamole}
\end{figure}

Figure \ref{whackamole} shows 20 steps of an iterative probabilistic calibration of a 1d $f(\x)$ that we can evaluate quickly (black line), to observations (solid grey line), with $\var_{\disc}$ and $\var_{\err}$ misspecified (dashed grey lines $\pm 3$ standard deviations) so that the true function does not come as close to the observations as the expert judgement indicates. Starting with an equally spaced 6 point design, a Gaussian process emulator is fitted for fixed correlation length (the mean function is the solid blue line, 2 standard deviation intervals are given by the blue dashed lines), and the posterior distribution $\pi(\best\mid\obs,\ens)$ overlaid (solid red line). We then evaluate $f(\x)$ at the maximum a posteriori estimate for $\best$, refit our Gaussian process, and compute the new posterior over $\best$ to produce the next plot. 

From panel 6 onwards, we see the issue with the terminal case for probabilistic calibration. Our posterior beliefs are highly peaked at one particular $\x$ value, yet evaluating the model there completely shifts the peak to a location for which we had near zero prior density. Each evaluation of the simulator, which for climate models may take weeks or months, shifts the posterior spike to an unexpected (a priori) part of parameter space. It is often not efficient to run expensive simulators, such as climate models, that require expert time to run and manage, one run at a time \citep{williamson2015exploratory}. The scientists that manage jobs on supercomputers, for example, require batches of runs that can be run in parallel. However, batch designs could be even worse here. Guided by the posterior density at each point, batch designs would be the near equivalent of one point at the MAP estimate, simply shifting the peak of the posterior to somewhere as yet unsampled.

Eventually, as we see from the bottom 4 panels, posterior uncertainty in $f(\x)$ is sufficiently reduced, and $\pi(\best\mid\ens, \obs)$ settles on the `least bad' value of $\x$, where $f(\x)$ is closest to the observations (though around 30 standard deviations away). For simulators with input spaces of much higher dimensions (the climate models we work with have typically specified 10-30 parameters to focus on, though these would be a subset of several hundred), we are unlikely to ever be able to reduce emulator uncertainty to the extent that the posterior spike settles over the least bad parameter setting. Hence, under an iterative procedure such as this, we would continue to chase the best input throughout parameter space, constantly moving the spike as in a game of whack-a-mole, until we run out of resources. 

Our illustration of the terminal case shows that though careful subjective prior information is required for model discrepancy in order to overcome the identifiability issues with the calibration model, if those judgements lead to a prior-data conflict via a terminal case, good calibration will not be possible, and it will take a great deal of resource (enough data to build a near perfect emulator everywhere) to discover this. It would seem more natural to first history match in order to check we are not in a terminal case, and, if not, perform a probabilistic calibration within NROY space as in \cite{salter2016comparison}.

%Arguably, if in the terminal case, history matching is the preferred methodology (a case argued by, for example, \cite{williamson2013history}). In the example, all of $\tspace$ is ruled out after 2 iterations. However, given a stopping rule, the same inference could be drawn by extending the probabilistic calibration methodology. We can also combine methods, by first history matching to check that the terminal case does not apply, before performing a probabilistic calibration within NROY space, . 

Whichever calibration method, or combination of them, is preferred, it is important to understand this terminal case, as we shall show that even for models that can reproduce observations exactly, tractable methods for calibrating high dimensional output can result in a terminal case analysis.

\section{Calibration with spatial output}\label{spatial}

For spatial fields, the most common approach to emulation and calibration involves projecting the model output onto a low-dimensional basis, $\bas$, and emulating the coefficients, so that fewer emulators are required \citep{bayarri2007computer, higdon2008computer, wilkinson2010bayesian, sexton2011multivariate} (although alternatives, such as emulating every grid box individually, have been applied, e.g. by \citet{gu2016parallel}).
\par
Writing the model output $f(\x_i)$ as a vector of length $\ell$, so that $\ens$ has dimension $\ell \times n$, the singular value decomposition (SVD) is used to give $n$ eigenvectors that can be used as basis vectors (equivalently, finding the principal components) \citep{higdon2008computer, wilkinson2010bayesian, sexton2011multivariate, chang2014probabilistic, chang2016calibrating}. For the size of model output typically explored using these methods, $\bas$ will not be of full rank as $n << \ell$. This means that while $\ens$ can be represented exactly by projection onto $\bas$, general $\ell$-dimensional fields will not have a perfect representation on $\bas$. As the majority of the variability in $\ens$ is usually explained by only the first few eigenvectors, the basis is truncated after $q$ vectors, giving a basis $\bas_q = (\bass_1, \ldots, \bass_q)$ of dimension $\ell \times q$, often chosen so that more than 95\% of $\ens$ is explained by $\bas_q$. Various rules-of-thumb are used dependent on the problem, e.g. \citet{higdon2008computer} truncate after 99\%, while \citet{chang2014probabilistic} use 90\%.
\par
In order to emulate the model, the runs are first centred by subtracting their mean, $\boldsymbol{\mu}$, from each  column of $\ens$, giving the centred ensemble $\ensc$ (we use the term ensemble to mean the collection of runs, as is common in the study of climate models). $\ensc$ is then projected onto the basis $\bas_q$, giving $q$ coefficients associated with each parameter choice:
\begin{equation} \label{projecteqn}
\textbf{c}(\x_{i}) = (\bas_q^{T}\bas_q)^{-1} \bas_q^{T} (f(\x_{i})  - \boldsymbol{\mu}).
\end{equation}
Given $q$ coefficients, a field of size $\ell$ is reconstructed via
\begin{equation} \label{reconeqn}
f(\x_{i}) = \bas_q \textbf{c}(\x_{i}) + \boldsymbol{\mu} + \boldsymbol{\epsilon},
\end{equation}
with $\boldsymbol{\epsilon} = \textbf{0}$ for $\x_i \in \textbf{X}$. Emulators for the coefficients of the first $q$ SVD basis vectors are then built:
\begin{equation} \label{basisemulators}
c_i(\x) \sim \text{GP}(m^*_i(\x), R^*_i(\x, \x; \boldsymbol{\phi})), \quad i = 1, \ldots, q. 
\end{equation}
Given these emulators, calibration can either be performed using the entire $\ell$-dimensional output, with emulator expectations and variances transformed to the $\ell$-dimensional space of the original field \citep{wilkinson2010bayesian}, or on its $q$-dimensional basis representation, with the observations projected onto this basis \citep{higdon2008computer}. 
\par
Calibration (via either history matching or probabilistic calibration) requires an informative prior process model for the spatial discrepancy, $\disc$. This could be a stationary process defined through a simple covariance function over the output dimensions, though a richer class of non-stationary process defined via kernel convolution is often used \citep{higdon1998process, chang2014probabilistic, chang2016calibrating}. These approaches specify a number of knots over the spatial field and define discrepancy to be a mixture of kernels around each of these knots. As with any calibration problem, however, strong prior information for discrepancy processes is essential to overcome identifiability issues, as discussed in Section \ref{terminal}. The way to include this information has been to fix the correlation parameters of the kernels and to have an infomative prior for their variances. With such a prior, a terminal case analysis is just as possible as for the 1D example we presented earlier.
\par
Suppose that the prior on the process is strong enough to overcome identifiability issues and is such that \textit{we don't have a terminal case}. When using a basis emulator to calibrate $f(\cdot)$, we may artificially induce a terminal case analysis, as reconstructions from coefficients on the basis are restricted to a $q$-dimensional subspace of $\ell$-dimensional space. Further, it will not be clear whether our analysis implies that the model is incapable of reproducing $\obs$, or that this was due to a poor basis choice. The SVD basis chooses the $q$-dimensional subspace that explains the maximum amount of variability in $\ens$ with the fewest number of basis vectors. This choice does not guarantee that important directions in $\ens$ that are consistent with $\obs$ are preserved.

\subsection{Illustrative example} \label{examplesection}
We illustrate this problem with an idealised example of a 6 parameter function $f(\x)$ (detailed in Section S1), with output given over a $10 \times 10$ grid. Observations, $\obs$, are given by a known input parameter setting, $f(\best)$, with $\mathrm{N}(0,\var_{\err})$ observation error added (given in (S2), with $\var_{\disc}$ defined in (S3)), so that a calibration exercise should be able to identify $\best$. In our example, the great majority of the input space leads to output that is biased away from $\obs$: the proportion of input space leading to output consistent with $\obs$ is around 0.01\%.
\par
The first panel of Figure \ref{toyobsmeanrecon} shows the observations, $\obs$, with a strong signal on the main diagonal. The second panel is the mean of the output field over a maximin Latin hypercube sample of size 60 in $\tspace$ (i.e. the mean of ensemble $\ens$). The strong signal in the ensemble is a biased version of $\obs$. In a climate context, this is analogous to the Gulf Stream being observed in the incorrect place in model output.
\begin{figure}[t]
\centering
\includegraphics[width = 0.65\linewidth]{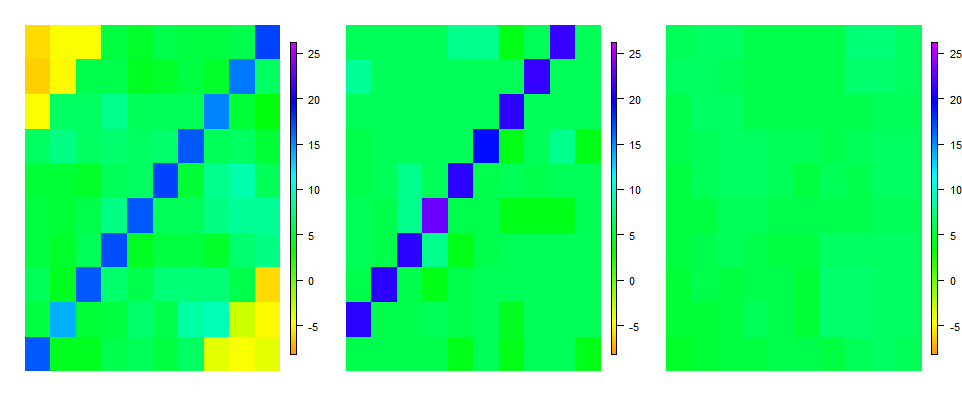}
\caption{Left: the observations, $\obs$, for our function. Centre: the ensemble mean. Right: the reconstruction of $\obs$ using the truncated SVD basis.}
\label{toyobsmeanrecon}
\end{figure}
\par
We calculate the SVD basis $\bas$ as described above. Over 95\% of the ensemble variability is explained by projection onto the first four basis vectors, which we refer to as the `truncated basis', $\bas_4$. If we project $\obs$ onto this basis and reconstruct the original field using these coefficients, via equations \eqref{projecteqn} and \eqref{reconeqn}, we obtain the field given by the third panel of Figure \ref{toyobsmeanrecon}: the distinctive pattern found in $\obs$ has been lost. That is, spatial calibration with $\bas_4$ would ultimately rule out parameter space that contained the true coefficients due to poor reconstruction, suggesting that, for reconstructions of the field using $\bas_4$, we are in the terminal case.
\par
Fitting Gaussian process emulators to the coefficients given by projection of $\ensc$ onto the four basis vectors, the expectation and variance at $\x$ is given by
\begin{align*}
\begin{split}
\mathrm{E}[\textbf{c}(\x)] = (\text{E}[c_1(\x)], \ldots, \text{E}[c_4(\x)])^T, \quad \mathrm{Var}[\textbf{c}(\x)] = diag(\text{Var}[c_1(\x)], \ldots, \text{Var}[c_4(\x)]).
\end{split}
\end{align*}
To probabilistically calibrate or history match on the original field, we require $\mathrm{E}[f(\x)]$ and $\mathrm{Var}[f(\x)]$ in terms of the coefficient emulators. These are
\begin{align*}
\begin{split}
\mathrm{E}[f(\x)] = \bas_q \mathrm{E}[\textbf{c}(\x)], \quad
\mathrm{Var}[f(\x)] = \bas_q \mathrm{Var}[\textbf{c}(\x)] \bas_q^T + \bas_{-q} \boldsymbol{\Sigma}_{-q} \bas_{-q}^T
\end{split}
\end{align*}
where $\bas_{-q}$ contains the discarded basis vectors, and $\boldsymbol{\Sigma}_{-q}$ is a diagonal matrix with the associated eigenvalues as the diagonal elements \citep{wilkinson2010bayesian}. 
\par
Calibrating in the coefficient subspace requires projection of $\obs$, $\var_{\disc}$ and $\var_{\err}$ onto $\bas_4$. For example, the implausibility in (\ref{mvimpl}) on the coefficients becomes 
\begin{equation} \label{coeffimpl}
\tilde{\mathcal{I}}(\x) = (\bas_q^T \obs - \mathrm{E}[\textbf{c}(\x)])^{T}(\mathrm{Var}[\textbf{c}(\x)] + \bas_q^T \var_{\disc} \bas_q + \bas_q^T\var_{\err} \bas_q)^{-1} (\bas_q^T \obs - \mathrm{E}[\textbf{c}(\x)]).
\end{equation}
Using the 0.995 value of the chi-squared distribution with 100 degrees of freedom to history match via \eqref{mvimpl}, we rule out the whole parameter space, $\tspace$, and so we are in the terminal case. Hence probabilistic calibration gives peaked prediction at the incorrect value of $\best$, consistent with the description given in Section \ref{terminal} (see SM section S1.1, Figures S3, S4). 

By history matching on the coefficients instead, using \eqref{coeffimpl}, and setting $\ibound$ using the chi-squared distribution with 4 degrees of freedom, we find an NROY space consisting of 3.8\% of $\tspace$. However, we rule out $58\%$ of the parameter space that was consistent with $\obs$, as the important directions for comparing the model to observations are not contained in $\bas_4$.
\par
%Even where parameter settings can be found that lead to coefficients consistent with the projection of $\obs$ onto the basis, there is no guarantee that the model output at these parameters is similar to $\obs$, as only $q << l$ directions are considered. These $q$ directions may not be most informative for $\obs$, being solely derived from a potentially-biased ensemble. Coefficients for discarded directions (i.e. low-eigenvalue basis vectors), or directions that are simply not given by the information we have available in the ensemble, may be far more important for determining whether the observations are given by $\x$.
\par
Whether we are calibrating on the original field, or on the coefficients, the `best' result we are able to find is that given by the reconstruction of $\obs$ with $\bas_4$, given in the final panel of Figure \ref{toyobsmeanrecon}. On the field, we are in the terminal case. On the coefficients, we are attempting to find runs that give coefficients that lead to this reconstruction, regardless of what happens in the directions we are interested in (i.e. the main diagonal pattern). Henceforth, we choose to focus on calibration on the field, as it compares all aspects of the observed output to the model, rather than a few summaries of it.

%Whether on the coefficients or the field, history matching for the idealised example was not able to overcome a poor basis choice.  The results for probabilistic calibration are consistent with this as we are in the terminal case (SI?), suggesting that an alternative basis is required in order to be able to identify the true $\best$.

\section{Optimal basis selection} \label{sectionbasis}

For calibration, there are two main requirements for a basis, $\B$, representing high dimensional output: being able to represent $\obs$ with $\B$ (a feature not guaranteed by principal component methods), and retaining enough signal in the chosen subspace to enable accurate emulators to be built for the basis coefficients (as principal components do). 
\par
A natural method for satisfying the first goal is to minimise the error given when the observations are reconstructed using $\B$. Define the reconstruction error, $\R_{\weight}(\B, \obs)$ via
\begin{equation} \label{reconerror}
\R_{\weight}(\B, \obs) = \lVert \obs - \textbf{B} (\textbf{B}^{T} \weight^{-1} \textbf{B})^{-1} \textbf{B}^{T} \weight^{-1} \obs \rVert_{\weight}.
\end{equation}
where $\lVert \textbf{v} \rVert_{\weight} = \textbf{v}^{T} \weight^{-1} \textbf{v}$ is the norm of vector $\textbf{v}$,
%\begin{equation} \label{reconerror2}
%
%\end{equation}
and $\weight$ is an $\ell \times \ell$ positive-definite weight matrix. By setting $\weight = \var_{\err} + \var_{\disc}$, $\R_{\weight}(\B, \obs)$ is analogous to \eqref{mvimpl}, and is the implausibility when we know the basis coefficients exactly (so that the emulator variance is $0$).
\par
As $\weight$ will not generally be a multiple of the identity matrix, the SVD projection from \eqref{projecteqn} is not appropriate for $\R_{\weight}(\cdot, \obs)$. Therefore, \eqref{projecteqn} becomes
\begin{equation}
\textbf{c}(\x_{i}) = (\B^{T} \weight^{-1} \B)^{-1} \B^{T} \weight^{-1} (f(\x_{i})  - \boldsymbol{\mu}),
\end{equation}
with this projection minimising the error in $\lVert \cdot \rVert_{\weight}$ (Section S2), hence the definition of the reconstruction error in \eqref{reconerror}.
\par
We present everything in full generality for positive definite $\weight$. Therefore, $\B$ is an orthonormal basis if $\B^T \weight^{-1} \B = \boldsymbol{\mathbb{I}}_{n}$. A basis with this property can be obtained using generalised SVD \citep{jolliffe2002principal}, with $\weight = \boldsymbol{\mathbb{I}}_{\ell}$ giving the usual SVD decomposition:
\begin{displaymath}
\ensc^T = \textbf{U} \textbf{D} \B^T, \quad \textbf{U}^T \textbf{U} = \boldsymbol{\mathbb{I}}_{n}, \, \B^T \weight^{-1} \B = \boldsymbol{\mathbb{I}}_{n}.
\end{displaymath}
As a measure of whether emulators can be built, we use the proportion of variability explained by projection of the ensemble onto each basis vector $\bb_k$, $\V_k(\B, \ensc)$, with
\begin{equation} \label{onevar}
\V_k(\B, \ensc) = \frac{\sum_{j = 1}^n \lVert \bb_k (\bb_k^T \weight^{-1} \bb_k)^{-1} \bb_k^T \weight^{-1} (f(\x_j) - \boldsymbol{\mu}) \rVert_{\weight} }{\sum_{j = 1}^n \lVert f(\x_j) - \boldsymbol{\mu} \rVert_{\weight}}.
\end{equation}
The proportion of ensemble variability explained by $\B$, $\V(\B, \ensc)$, is
\begin{equation} \label{generalvar}
\V(\B, \ensc) = \frac{\sum_{j = 1}^n \lVert \B (\B^T \weight^{-1} \B)^{-1} \B^T \weight^{-1} (f(\x_j) - \boldsymbol{\mu}) \rVert_{\weight} }{\sum_{j = 1}^n \lVert f(\x_j) - \boldsymbol{\mu} \rVert_{\weight}}.
\end{equation}
The SVD basis maximises $\V_k(\B, \ensc)$ for each $k$, given the previous vectors and subject to orthogonality.
\par
Prior to building emulators and performing calibration for a given basis, we can assess whether we are in the terminal case or not. For history matching threshold $T$, if
%\begin{equation} \label{reconbound}
$\R_{\weight}(\B, \obs) > \ibound$
%\end{equation}
then we are in the terminal case on $\B$, and would even rule out values of $\best$ that reproduce $\obs$ exactly. If $\R_{\weight}(\B, \obs) > \ibound$ for some $\{\B, \weight\}$, we may view $\weight$ as having been misspecified, as in the terminal case described in Section \ref{terminal}. However, we may also have under-explored the output dimension of $f(\cdot)$, so that $\B$ does not allow us to get close enough to $\obs$. We revisit this test in the context of optimal basis choice in Section \ref{algorithm_sect}.

\par
Figure \ref{toysvdvarmse} compares $\V(\cdot, \ensc)$ and $\R_{\weight}(\cdot, \obs)$ for the example of Section \ref{spatial}. We refer to plots of this type as VarMSE plots. The red line represents $\R_{\weight}(\B_k, \obs)$, and the blue line shows $\V(\B_k, \ensc)$, for each truncated basis, $\{ \B_k \}^n_{k = 1}$. The vertical dotted line indicates where the basis is truncated if we wish to explain 95\% of the ensemble variability, and the horizontal dotted line represents the history matching bound, $\ibound$. The solid horizontal line is equal to $\R_{\weight}(\B, \obs)$.
\begin{figure}[t]
	\centering
	\includegraphics[width = 0.45\linewidth]{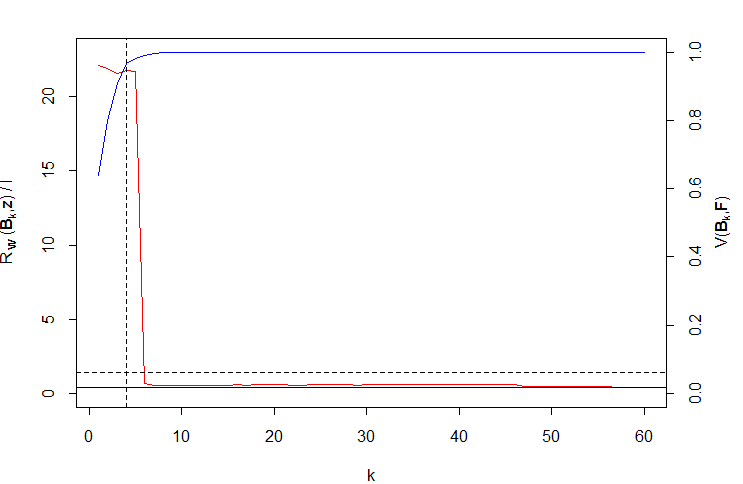}
	\caption{A plot showing how the reconstruction error (red) and proportion of ensemble variability explained (blue) change as the SVD basis is increased in size, for $\weight = \var_{\err} + \var_{\disc}$.}
	\label{toysvdvarmse}
\end{figure}
For the SVD basis in our example, we see that $\R_{\weight}(\cdot, \obs)$ is large (compared to $\ibound$) until $k = 6$, where the error decreases below the threshold, indicating that the sixth basis vector contains patterns that are important for explaining $\obs$. As further basis vectors are added, $\R_{\weight}(\cdot, \obs)$ continues to decrease, suggesting that patterns relevant for representing $\obs$ are in fact included in $\B$ for this example. However, the later basis vectors explain low percentages of the variability in the ensemble, with the low signal to noise ratio of projected coefficients making accurate emulation impossible. If we could emulate the coefficients for the 5th and 6th basis vectors, we would more accurately represent $\obs$, although this rapid decrease in the reconstruction error is a feature of our example, rather than a general property of the SVD basis, and therefore we still truncate at 95\% for illustrative purposes.
\par
The SVD basis aims to maximise the blue line for each basis vector added, whereas, for calibration, we require the red line to be below $T$. The problem of basis selection for calibration is one of trading off these two requirements, reducing $\R_{\weight}(\cdot, \obs)$ while ensuring that each $\V_k(\cdot, \ensc)$ is large enough to enable emulators to be built. Given that the full SVD basis may contain information and patterns that allow the observations to be more accurately represented, the information contained in this basis may be combined in such a way that the resulting basis is suitable for calibration, with important low-order patterns blended with those that explain more of the ensemble variability.

\subsection{Rotating a basis}
Performing a rotation of an ensemble basis $\B$ using an $n \times n$ rotation matrix, $\rot$, rearranges the signal from the ensemble, potentially allowing the new truncated basis to be a better representation of $\obs$. A general $n \times n$ rotation matrix $\rot$ can be defined by composing $n(n-1)/2$ matrices that give a rotation by an angle around each pair of dimensions \citep{murnaghan1962unitary}. Our goal is to find $\rot$ such that $\B \rot$ minimises $\R_{\weight}((\B \rot)_q, \obs)$, subject to constraints on $\V_k(\cdot, \ensc)$ that allow the projected coefficients to be emulated.
\par
To directly define a rotation matrix $\rot$ via optimisation requires a large number of angles to be found, even when the ensemble size is small. Instead, we take an iterative approach, selecting new basis vectors sequentially while minimising $\R_{\weight}(\cdot, \obs)$ at each step, in such a way that guarantees that the resulting basis is an orthogonal rotation of the original basis.

Given $p < n$ basis vectors, $\B_p = (\bb_1, \ldots, \bb_p)$, we define the `ensemble residual' as
\begin{displaymath}
\resens = \ensc - \B_p (\B_p^{T} \weight^{-1} \B_p)^{-1} \B_p^{T} \weight^{-1} \ensc
\end{displaymath}
This represents the variability in the ensemble not explained by $\B_p$. Define the `residual basis', $\resbas$, to be the matrix containing the right singular vectors of $\resens$. The residual basis gives basis vectors that explain the remaining variability in $\ensc$, given vectors $\B_p$.

%Combining the first $(n - p)$ vectors of this basis with $\B_p$, we have a basis of rank $n$, explaining all of the variability in $\ensc$.

%The vectors in $\B_p$ may be elicited physical patterns (although there is no guarantee that each $\V_k(\B_p, \ensc)$ is large enough).

\subsection{The optimal rotation algorithm}\label{algorithm_sect}

Given an orthogonal basis $\B$ for $\ensc$ with dimension $\ell \times n$; a positive definite $\ell \times \ell$ weight matrix $\weight = \var_{\disc} + \var_{\err}$; a vector $\textbf{v}$, where $v_i$ is the minimum proportion of the ensemble variability to be explained by the $i^{th}$ basis vector; the total proportion of ensemble variability to be explained by the basis $v_{tot}$; and a bound $T$ (usually that implied by history matching, $\ibound = \chi^2_{\ell, 0.995})$, we find an optimal basis for performing calibration as follows:
\begin{enumerate}
	\item If $\R_{\weight} (\B, \obs) > T$, stop and revisit the specification of $\weight$, or add more runs to $\ensc$. Else set $k = 1$.
	\item Let $\bas_k^* = (\bass_1^*, \ldots, \bass_{k - 1}^*, \B \rott_k)$ and set
	\begin{displaymath}
	\rott_k^* = arg min_{\rott_k} \R_{\weight}(\bas_k^*, \obs)
	\end{displaymath}
	such that $\V_k(\bas_k^*, \ensc) \geq v_k.$
	Define the new normalised vector as
	\begin{displaymath}
	\bass^*_k = \frac{\B \rott_k^*}{\sqrt{\lVert \B \rott_k^* \rVert_{\weight}}},
	\end{displaymath}
	and set $\bas_k^* = (\bass_1^*, \ldots, \bass_{k-1}^*, \bass_k^*)$.
	\item Find the residual basis given $\bas_k^*$, $\B_{\epsilon}^k$, and form the orthogonal rank $n$ basis
	\begin{displaymath}
	\bas^* = (\bas_k^*, [\B_{\epsilon}^k]_{n-k}).
	\end{displaymath}
	\item Define $q \geq k$ as the minimum value satisfying $\V(\bas_q^*, \ensc) \geq v_{tot},$ 
	where $\bas_q^*$ represents the first $q$ columns of $\bas^*$. If $\R_{\weight} (\bas_q^*, \obs) < T,$
	then stop, and return $\bas_q^*$ as the truncated basis for calibration. Else, set $k = k + 1$ and $\B = [\B_{\epsilon}^k]_{n-k}$, and return to step 2.
\end{enumerate}
Prior to applying the algorithm, we must specify an initial basis, $\B$, a weight matrix, $\weight$, and the parameters $v_{tot}$ and $\textbf{v}$ to control the amount of variability explained by each basis vector. We use the SVD basis (with respect to $\weight$) for $\B$, however other choices are possible, e.g. we could apply Gram-Schmidt to the ensemble itself and rotate this, or apply a different scaling to the SVD basis.
\par
At each step, our algorithm selects the linear combination of a given basis that minimises $\R_{\weight}(\cdot, \obs)$, subject to explaining a given percentage of ensemble variability, and given any previously selected basis vectors. If the defined truncation $\bas_q^*$ satisfies $\R_{\weight}(\bas_q^*, \obs) < T$, then the algorithm terminates, as standard residual variance maximising basis vectors no longer lead to a terminal case analysis. We allow a basis to be identified that satisfies our two goals: we do not rule out $\obs$, and have coefficients that can be emulated, if $\textbf{v}$ is set appropriately. To optimise for $\rott_k$, we use simulated annealing \citep
{gensa2013}, although any optimisation scheme that converges could be used.
\par
The check in step 1 of our algorithm is due to the following result (proved in S2):
\begin{result}[Invariance of $\R_{\weight}(\cdot, \cdot)$ to rotation] For a rotation matrix $\rot$ of dimension $k \times k$, and a set of basis vectors $\B = (\bb_1, \ldots, \bb_n)$, we have
	\begin{equation}
	\R_{\weight}(\B_k, \obs) = \R_{\weight}(\B_k \rot, \obs), \quad k = 1, \ldots, n
	\end{equation}
\end{result}
Regardless of the rotation that is applied to $\B$, we cannot reduce the reconstruction error below that given by the full basis originally. However, because the SVD basis is always truncated prior to this minimum value being reached, we can search for a rotation that rearranges the information from the SVD basis in such a way that satisfies
\begin{equation}
\R_{\weight}((\B \rot)_q, \obs) << \R_{\weight}(\B_q, \obs), 
\end{equation}
incorporating important, potentially low-order, patterns into the $q$ basis vectors that we emulate. Hence step 1 of the algorithm provides an important test as to whether our ensemble and uncertainty assessment, $(\ens, \weight)$, are sufficient to avoid a terminal case analysis, and shows when a rotation exists, up to the choice of $\textbf{v}$.
\begin{theorem} \label{theoremrotation}
$\bas^*$ in step 3 of the optimal rotation algorithm is an orthogonal rotation of $\B$.
\end{theorem}
The results and proofs required, and the proof of Theorem \ref{theoremrotation} itself, are found in Section S2. Given that $\B$ passes step 1 of the algorithm, existence of an optimal rotation depends on the choice of $\textbf{v}$:
\begin{theorem}
At the $k^{th}$ iteration of the optimal rotation algorithm, given an orthogonal $\bas_{k-1}^*$ that satisfies $\V_j(\bas_{k-1}^*, \ensc) \geq v_j, \, j = 1, \ldots, k-1$, $\exists \, \bass_k^* \, _\bot \, \bas_{k-1}^*$ with $\V(\bass_{k}^*, \ensc) \geq v_k$ and $\R_{\weight}(\bas_k^*, \obs) \leq \R_{\weight}(\tilde{\bas}_k, \obs) \leq \R_{\weight}(\bas_{k-1}^*, \obs)$, for 
$\bas_k^* = (\bas_{k-1}^*, \bass_k^*), \tilde{\bas}_k = (\bas_{k-1}^*, \tilde{\bass}_k)$, and $\V(\tilde{\bass_{k}}, \ensc) \geq v_k \, \forall \tilde{\bass_k} \, _\bot \, \bas_{k-1}^*$ $\, \iff \, \V_1(\B_{\epsilon}^{k-1}, \ensc) \geq v_k$. In this case the algorithm converges to $\bass_k^*$.
\end{theorem}

\begin{proof}
By construction, $\V_1(\B_{\epsilon}^{k-1}, \ensc) = \max_{j} \V_j(\B_{\epsilon}^{k-1}, \ensc) = \max \V(\boldsymbol{\epsilon}, \ensc) \,$ $\forall \, \boldsymbol{\epsilon} \in span \{ \resens^{k-1} \}$. Hence if $\V_1(\B_{\epsilon}^{k-1}, \ensc) < v_k, \, \not\exists \, \bass_k^* = \B_{\epsilon}^{k-1} \rott_k$ such that $\V(\bass_k^*, \ensc) \geq v_k$.
  
\noindent If $\V_1(\B_{\epsilon}^{k-1}, \ensc) \geq v_k \implies \exists \bass_k^* = \B_{\epsilon}^{k-1} \rott_k$ with 

i) $\V(\bass_k^*, \ensc) \geq v_k$,

ii) $\bass_k^* \, _\bot \bas_{k-1}^*$ (by Theorem \ref{theoremrotation}), 

iii) $\R_{\weight}(\bas_k^*, \obs) \leq \R_{\weight}(\bas_{k-1}^*, \obs)$: let $\cc_{k-1}^* = ((\bas_{k-1}^*)^T \weight^{-1} \bas_{k-1}^*)^{-1} (\bas_{k-1}^*)^T \weight^{-1} \obs$ and $\cc_k^* = ((\bas_k^*)^T \weight^{-1} \bas_k^*)^{-1} (\bas_k^*)^T \weight^{-1} \obs$ be the coefficients given by projecting $\obs$ onto $\bas_{k-1}^*$ and $\bas_k^*$ respectively. Let $\cc^* = (\cc_{k-1}^*, 0)$, then
\begin{displaymath}
\R_{\weight}(\bas_{k-1}^*, \obs) = \lVert \obs - \bas_{k-1}^* \cc_{k-1}^* \rVert_{\weight} = \lVert \obs - \bas_k^* \cc^* \rVert_{\weight} \geq \lVert \obs - \bas_k^* \cc_k^* \rVert_{\weight} = \R_{\weight}(\bas_k^*, \obs).
\end{displaymath}
as by construction $\cc_k^*$ minimises the reconstruction error in the $\weight$ norm.

 Finally, $\R_{\weight}(\bas_k^*, \obs) \leq \R_{\weight}(\tilde{\bas}_k, \obs) \, \forall \, \tilde{\bass}_k = \B_{\epsilon}^{k-1} \tilde{\rott}_k$ with $\V(\tilde{\bass}_k, \ensc) \geq v_k$ (by convergence of the optimiser, e.g. \citet{aarts1985statistical} for simulated annealing).
\end{proof}

In practice, when applying our algorithm to high dimensional model output, we have found that only a small number (three or fewer) of iterations have been required, hence $\textbf{v}$ often has a low dimension. The values of $\textbf{v}$ required will depend on the problem, with a different approach required when a small number of vectors explain the majority of the ensemble, compared to when a large proportion of the variability is spread across many SVD basis vectors. In the former case, the values of $\textbf{v}$ may be relatively high, whilst in the latter they can be lower, relative to the proportion explained by the equivalent SVD basis vectors. A reasonable approach is to initially set $\textbf{v}$ as half of the proportion explained by the corresponding SVD basis vectors, reducing these further if the resulting $q$ is too large. As Theorem 2 shows, it is possible to set $\textbf{v}$ in such a way that the algorithm is unable to find a suitable basis. If we cannot find a $k^{th}$ basis vector that satisfies the variability constraint, given $\bas_{k-1}^*$, then a basis doesn't exist for this choice of $\textbf{v}$, and the specification needs revisiting: either $v_k$ needs to be decreased, or an earlier constraints needs relaxing.
\par
The choice $v_{tot}$ is also a concern for the standard UQ approaches based on principal components. In our experience, using similar rules (e.g. 95\% or 99\%) to the SVD applications leads to 0-2 extra basis vectors required.
\par
In an application, it may be desirable to include certain physical patterns, deemed to be important, in our basis $\B$, which may not lie within the subspace defined by $\ensc$. In this case, if we have $p$ selected physical vectors, $\B_p = (\bb_1, \ldots, \bb_p)$, combining these with the first $n - p$ vectors of the residual basis will not necessarily explain all of the variability in $\ensc$. The algorithm may be applied to the $n + p$ vectors given by the physical vectors and the full residual basis, giving a rotation of this space rather than of $\ensc$. As truncation occurs after the majority of variability of $\ensc$, $v_{tot}$, is explained, the resulting truncated basis, while not strictly a rotation of the subspace defined by $\ensc$, will exhibit similar qualities, and may be superior for representing $\obs$, if important physical patterns can be emulated when combined with signal from the ensemble.
\par
To perform the algorithm with basis vectors from outside the subspace defined by $\ensc$, rather than finding linear combinations of the residual basis at step $k > 1$, $\B = (\B_p, \resbas)$ is used at each step, with orthogonality imposed after each new basis vector has been selected, via Gram-Schmidt (as by Result S3, applying Gram-Schmidt does not affect $\R_{\weight}(\cdot, \obs)$).

\subsection{Idealised example continued}

We now apply the optimal rotation algorithm to the example of Section \ref{spatial}. We set $\textbf{v} = (0.4, 0.1, 0.1)$, $v_{tot} = 0.95$, and $\B$ as the SVD basis, with the projection of \eqref{projecteqn} used for consistency with Section \ref{examplesection}, to show that rotation fixes the described problems. One iteration of the algorithm finds a basis that satisfies $\R_{\weight} (\bas_q^*, \obs) < T$, with $q = 5$ (i.e. we need the first 4 vectors of the residual basis so that $\bas_q^*$ explains at least 95\% of $\ensc$). 
\par
The reconstruction of $\obs$ with this basis, and associated VarMSE plot, are shown in Figure \ref{wave1reconvar}. Now, our basis allows us to accurately represent $\obs$ with the leading vectors, as the important patterns from low-order eigenvectors have been combined with the leading patterns (hence an additional vector being required to explain more than 95\% of $\ensc$).
\par
Performing history matching as before, and using the reconstructions of the original fields rather than the coefficients, we find that 31.5\% of $\tspace$ is now in NROY space (Figure S5). Performing our previous check on the accuracy of the match, we find that no runs consistent with $\obs$ have been ruled out.
\par
As we are no longer in the terminal case, we perform probabilistic calibration on the field. The posterior densities found by calibrating on $\bas_q^*$ are shown in Figure S3, with the average simulator output given by samples from this posterior in the first plot of Figure \ref{calw123mean}. While the samples here are not consistent with $\obs$, as the off-diagonal is too strong, we have been able to identify runs where there is signal on the main diagonal. This is because the rotated basis allows for this direction of the output space to be searched. The limited signal in the important directions from $\ens$ has been extracted and used to guide calibration.
\par
\begin{figure}[t]
\centering
\includegraphics[width = 0.55\linewidth]{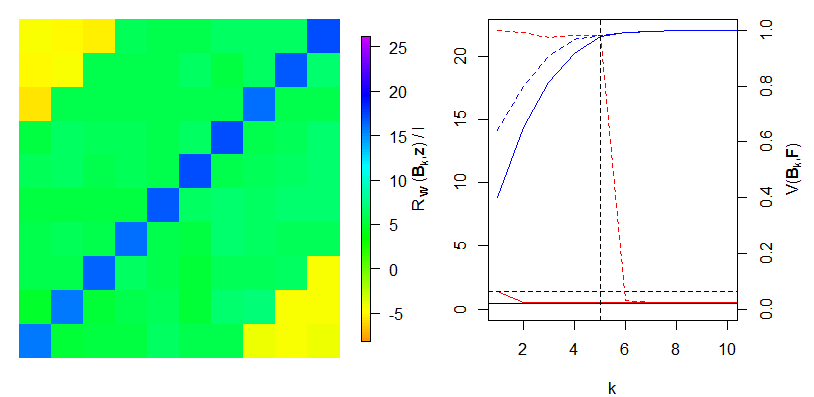}
\caption{The reconstruction of $\obs$ using the truncated basis $\bas^*_q$, and the VarMSE plot for this basis, with the truncated SVD basis given by the red and blue dotted lines.}
\label{wave1reconvar}
\end{figure}
We continue the calibration by running a new design within NROY space. This new design should contain more signal in the direction of $\obs$, and hence it should be possible to find a rotation that reduces $\R_{\weight} (\cdot, \obs)$ further than at the previous wave. We select 60 points from the wave 1 NROY space and run $f(\cdot)$ at these points to give the wave 2 ensemble. We perform a rotation, and emulate and calibrate using the wave 2 ensemble. History matching reduces NROY space to 3.1\% of $\tspace$ (Figure S7). If we instead perform probabilistic calibration, with zero density assigned to regions outside of the wave 1 NROY space, we find the average output field in the 2nd plot of Figure \ref{calw123mean} (posteriors in Figure S6).
\par
These results represent a large improvement over performing only one wave. We have ruled out the majority of $\tspace$, allowing future runs to be focussed in this region. Probabilistic calibration is more accurate, with samples containing a strong diagonal, as with $\obs$.
\par
Repeating the process, our wave 3 ensemble contains patterns more consistent with $\obs$ than in previous waves, and hence the truncated SVD basis does not rule out the reconstruction of $\obs$, and no rotation is required. Following emulation for this basis, history matching leads to an NROY space consisting of 2\% of $\tspace$ (Figure S8). Probabilistic calibration (in the wave 2 NROY space) gives the average output in Figure \ref{calw123mean} (posteriors in Figure S6), showing that our samples are now consistent with $\obs$.
\par
\begin{figure}[t]
\centering
\includegraphics[width = 0.65\linewidth]{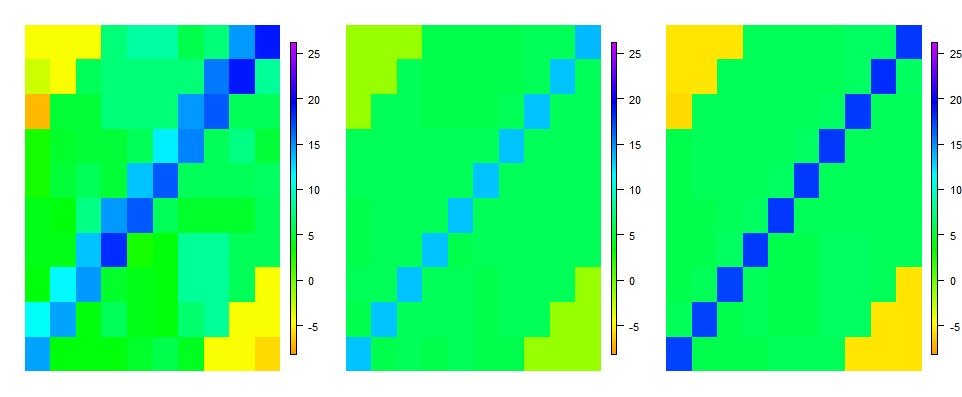}
\caption{The mean output $f(\x)$ for samples of $\x$ from the probabilistic calibration posterior, for the wave 1 rotated basis, the wave 2 rotated basis and the wave 3 SVD basis.}
\label{calw123mean}
\end{figure}

\section{Application to tuning climate models}  \label{canhm}

In this section, we will demonstrate that optimal rotation is an important and necessary tool if attempting to perform UQ for climate model tuning. However, as we will discuss, climate model tuning is not a solved problem, and it would be of limited value to simply show how calibration with our method can lead to an improvement in one output field over the standard methods, without necessarily improving the whole model or addressing the concerns of the community. We will motivate our discussion using the current Canadian atmosphere model, CanAM4.  

CanAM4 is an Atmospheric Global Climate Model, which integrates the 
primitive equations on a rotating sphere employing a spherical-harmonic 
spatial discretization truncated triangularly at total wavenumber 63 
(T63), with 37 vertical levels spanning the troposphere and stratosphere 
\citep{von2013canadian}.  CanAM4 has a large number of adjustable, 
`free', parameters of which 13 will be varied here. The climatological influence of each set of free 
parameters is determined from 5-year `present-day' integrations with 
prescribed sea-surface temperatures and sea-ice.  Model output is 
represented on the `linear' 128$\times$64 Gaussian grid corresponding to the 
model's T63 spectral resolution.

There are many output fields that must be checked for consistency with the observed climate when tuning the parameters of a climate model (in the case of CanAM4 there are more than 20). Here we focus on just 3 2D fields: vertical air temperature (TA), the top of the atmosphere radiative balance (RTMT, $Wm^{-2}$) and the cloud overlap percentage (CLTO). For RTMT and CLTO, the output is given over a  longitude-latitude grid, so that $\ell = 8192$. TA is the temperature averaged over longitude for each latitude and vertical pressure level so that $\ell = 2368$. There is also a temporal aspect to the output, with monthly values for every grid box; however, we remove this here by averaging over 5 years of June, July, August (JJA) output.
\begin{figure}
\centering
\begin{subfigure}{.33\textwidth}
  \centering
  \includegraphics[width=1.1\linewidth]{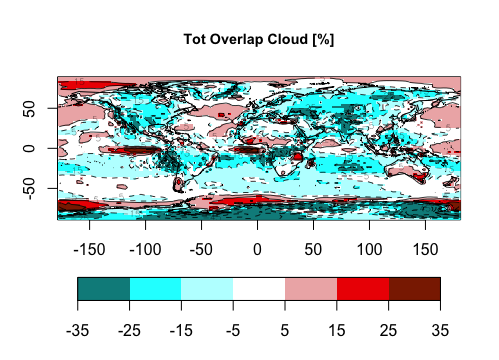}
\end{subfigure}%
\begin{subfigure}{.33\textwidth}
  \centering
  \includegraphics[width=1.1\linewidth]{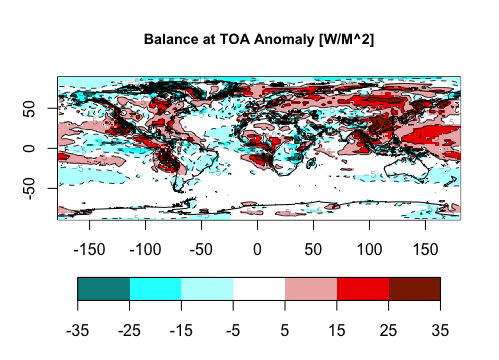}
\end{subfigure}
\begin{subfigure}{.33\textwidth}
  \centering
  \includegraphics[width=0.9\linewidth]{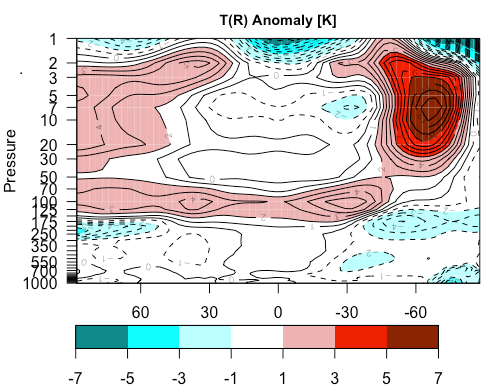}
\end{subfigure}
\caption{The standard CanAM4 anomaly field for a) CLTO, b) RTMT and c) TA.}
\label{anomalies}
\end{figure}
\par
When evaluating and tuning the model, spatial anomaly plots are routinely examined to see how the model compares with observations. An anomaly plot shows the difference between the model and the observations with a blue-white-red colour scale set such that blue is `too negative', white is `alright' and red is `too positive'. So, for example, in a temperature anomaly plot, red areas show where the model is too warm (for the modellers) compared to observations. Figure \ref{anomalies} shows anomaly plots for CLTO, RTMT, and TA for the standard configuration of CanAM4, with the colour scales representing the standard colours used by the modellers when tuning the model. 

A goal of tuning is to try to reduce or remove biases that are visible from these plots. Yet equally important is to learn which biases cannot be removed simply by adjusting the model parameters. This is the search for `structural errors' in the model (what statisticians would call model discrepancies). Structural errors indicate that there are flaws with individual parametrisations, or with the way they interact, that cannot be fixed by tuning. Understanding what these structural errors are so that they might be addressed either as part of this phase of development or for the next is one of the major goals of tuning \citep{williamson2015bias}. However, joint estimation of model discrepancy variances and model parameters is not possible without strong prior information \citep{brynjarsdottir2014learning} due to lack of identifiability. 

When working with CanAM4 then, our goal is to use history matching with a `tolerance to error' discrepancy variance \citep{williamson2015bias, williamson2017tuning} that aims to reduce the size of NROY space, so that, ultimately, in a calibration exercise we have strong prior information about $\best$ and some structured information on discrepancy. A formal methodology for achieving this is beyond the scope of this paper. However, we will demonstrate that optimal rotation is a crucial component for any attempt of this nature.

\begin{figure}[t]
\centering
\includegraphics[width = 0.7\linewidth]{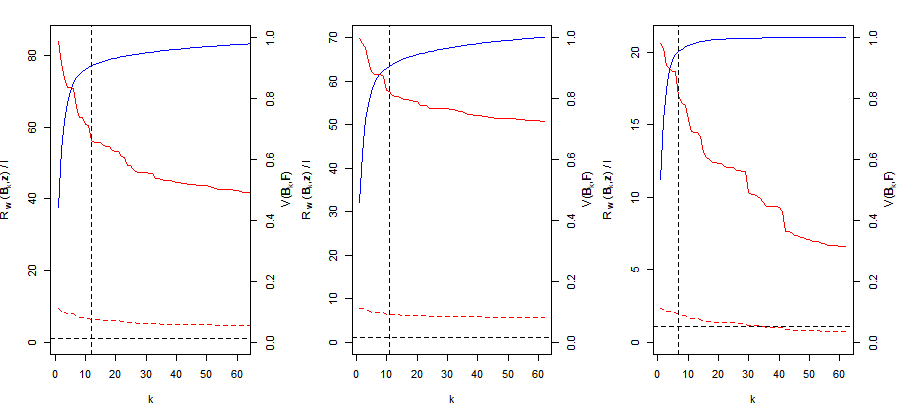}
\caption{VarMSEplots for a) CLTO b) RTMT and c) TA, with $\weight$ based on 1SD (dotted line) and 3SD (solid line). The dotted horizontal line indicates $\ibound$.}
\label{canamw1var}
\end{figure}

We designed 62 runs of CanAM4, varying $13$ parameters and using a $k$-extended Latin Hypercube \citep{williamson2015exploratory}. Figure \ref{canamw1var} shows VarMSE plots for the output fields CLTO, RTMT and TA for this ensemble. The weight matrix $\weight$ used for the reconstructions represents our tolerance to model error (we discussed its correspondence to model discrepancy ($\weight =\var_{\disc}+\var_{\err}$) in Section 4), and the red lines in these plots represent 2 alternatives based on interpreting the colour scales pertaining to the white regions in Figure \ref{anomalies}. We assume that the white region represents 3 standard deviations (solid red line) and 1 standard deviation (dashed red line), and set a diagonal $\weight$ accordingly. The solid red lines on each plot indicate that we have a terminal case analysis under the small model discrepancy. 

The larger discrepancy indicates a terminal case in CLTO and RTMT, and that 35 basis vectors would be enough to adequately reconstruct TA. However, the blue line in the TA plot shows that there is so little ensemble signal on the basis vector coefficients after arguably 20 (or fewer) basis vectors, that calibration on 35 basis vectors is not possible. If discrepancy were increased (an operation that involves scaling the red line until it lies below the bound $T$ represented by the dashed horizontal line in the plots), all 3 panels demonstrate that the reconstruction error under SVD decreases too slowly, so that a large number of basis vectors, each with coefficients that are increasingly difficult to emulate due to the decreasing ensemble signal, would be required to avoid a terminal case analysis.

Suppose model discrepancy $\var_{\disc} >> \var_{\err}$ so that we can consider $\weight = \var_{\disc}$ in the following. In order to use optimal rotation, we require $\weight$ such that $\R_{\weight} (\B, \obs) < T$, which is not true under our specification above for RTMT and CLTO. If we really believed our $\var_{\disc}$ represented the climate model's ability to reproduce observed climate, then this indicates that we need a larger ensemble in order to explore the model's variability. In that case, it may be desirable to follow a procedure like the one we present here to design these runs. 

In our case, we believe it is clear that we have misspecified model discrepancy. In fact, we used a place-holder tolerance to error, so this analysis indicates that we are not tolerant enough to model error (at this stage). To explore model discrepancy, we first perform a rotation under the $\weight$ given above, using the algorithm without step 1 in order to find $\R_{\weight}(\bas^*_q, \obs)$ as close to the reconstruction error of the untruncated SVD basis as possible, for small $q$ and whilst retaining emulatability by setting $\textbf{v}=(0.35,0.1,0.1)$ (as 3 rotated vectors is enough). Given this rotation, we then set
\begin{equation} \label{newdisc}
\var_{\disc} = \frac{\R_{\weight}(\bas_q^*, \obs)}{b} \weight, \quad b = \chi^2_{\ell, j}
\end{equation}
where $j < 0.995$ is a tuning parameter. This ensures that when reconstructed with the new basis, the observations will not be ruled out, and hence we can identify an NROY space likely to contain runs as consistent with $\obs$ as possible, given the limited information we have with 62 ensemble members. This has the effect of `scaling' the reconstruction error for the rotated basis seen in Figure \ref{canamw1rot} below the horizontal dotted line at the point the basis is truncated. For our fields, we set $j = 0.95$ for RTMT, and $j = 0.68$ for the others (as $j = 0.95$ ruled out all of $\tspace$ for CLTO and TA).

\begin{figure}[t]
\centering
\includegraphics[width = 0.7\linewidth]{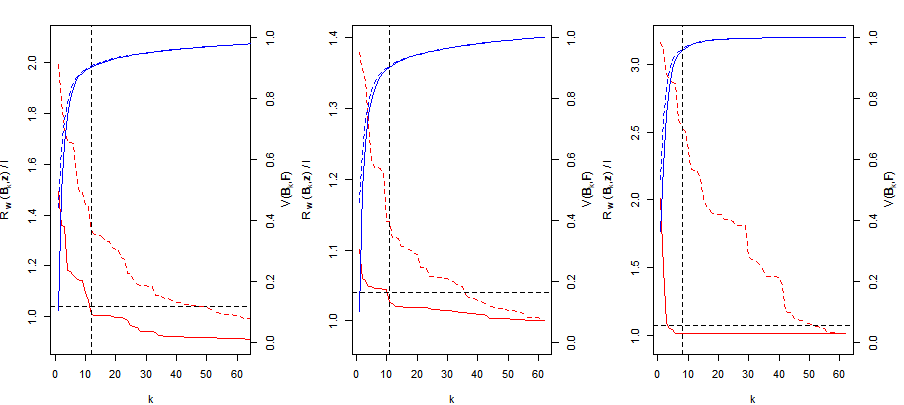}
\caption{VarMSEplots for CLTO, RTMT and TA, with $\weight = \var_{\disc}$, with rotated basis (solid lines) and SVD basis (dotted lines). The dotted horizontal line indicates $\ibound$.}
\label{canamw1rot}
\end{figure}

%\subsection{History matching} \label{canhm}

%Due to the size of the output, we are not able to evaluate the implausibility on the field for enough $\x$, and hence match using the coefficients instead. We can however calculate the implausibility on the full field for a small sample of $\x$ values, and use this to set the bounds $\ibound_{TA}, \ibound_{CLTO}, \ibound_{RTMT}$, giving a more accurate NROY space than simply using the chi-squared bound \citep{thesis}.

We define NROY space as runs where $\x$ is not ruled out using \eqref{mvimpl} for each of TA, CLTO and RTMT. This NROY space consists of 0.9\% of $\tspace$. We then design and run a new 50-member ensemble within this NROY space (discussed in \citet{thesis}, Section 6.3.5).
\par
Upon inspection of the TA field for this wave 2 ensemble, we observe that every run contains the previously found strong warm bias in the Southern Hemisphere (Figure S9). As our optimal basis choice permitted the search for runs not containing this structural bias, these results are evidence that this may be a structural error. In practice, how much evidence is required before the modellers are convinced that a particular bias is structured or not is a climate modelling decision. Certainly, we could repeat our wave 1 procedure within the current NROY space and run a wave 3 and so on. This has the benefit of increasing the density of points in $\tspace$ and the accuracy of emulators in key regions of $\tspace$, thus insuring against possible `spikes' in the model input space that would correct the bias.

Assuming our modellers were convinced to treat this feature as a structural bias, we demonstrate an approach to include this information within the iterative calibration procedure. We first revisit the specification of the TA discrepancy, selecting the region with the common warm bias shown in Figure S10, deemed to be a structural error, and increasing $\var_{\disc}$ for the grid boxes in this region. To do this, we set $\weight$ as a diagonal matrix with 100 for the grid boxes in this region, and 1s elsewhere on the diagonal (note this is one possible choice. We might, instead, increase the correlation between these gridboxes in $\weight$ in addition).
\par
We re-define the wave 1 NROY space so that it only depends on CLTO and RTMT (consisting of 41.4\% of $\tspace$), and then include NROY wave 1 runs with the wave 2 ensemble when rotating and building emulators for wave 2. For TA, the optimal rotation algorithm is applied using the newly-designed $\weight$, with the discrepancy $\var_{\disc}$ defined via \eqref{newdisc}, to ensure that $\obs$ is not ruled out ($\weight$ reflected our beliefs about the structure, not the magnitude, of $\var_{\disc}$). History matching using the wave 2 bases and emulators leads to an NROY space containing 0.03\% of $\tspace$. Plots illustrating this NROY space for six of the more active parameters are shown in Figure \ref{w2nroycanada}. We see that the regions with the greatest density of points in NROY space are generally found towards the edges of the parameter ranges. From the lower left plots, it is easier to identify relationships between some of the parameters, e.g. CBMF generally needs to be high while UICEFAC needs to be towards the centre of its allowable range.

\begin{figure}[t]
\centering
\includegraphics[width = 0.42\linewidth]{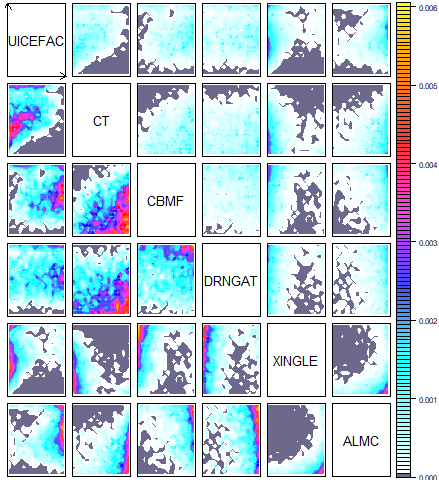}
\caption{Plot showing the composition of the wave 2 NROY space for 6 of the parameters of CanAM4. For each cell of a pairwise plot, a large sample is taken over the remaining parameters and the proportion of space that is not ruled out is calculated. The lower left gives the same plots, with scales set for each individual plot to show more structure.}
\label{w2nroycanada}
\end{figure}

The calibration of climate models, or even simply the search for structural biases, is a massive undertaking, and a full tuning is well beyond the scope of this paper. Each small ensemble of CanAM4 required 2 days of super-computer time to run and, in reality, the modellers routinely check over $20$ spatial fields (and a host of other metrics) when tuning the model. UQ can help with the tuning procedure in providing tools (emulators) that allow $\tspace$ to be explored much more quickly than is currently possible at the modelling centres. However, as this application has demonstrated, using the off-the-shelf methods based on the SVD basis does not work for tuning in general. It did not work for the 3 fields we showed here, nor have the authors ever found climate model output for which the problems we identified here were not present. Our application demonstrates the optimal rotation algorithm as an effective solution to quickly find bases without these issues for calibration.

\vspace{-0.2in}
\section{Discussion}

In this paper, we highlighted the issue of terminal case analyses for the calibration of computer models. A terminal case analysis occurs when the prior assessment of model discrepancy is inconsistent with the computer simulator's ability to mimic reality, and leads either to useless and incorrect posterior distributions (using the fully Bayesian procedure) or ruling out all of parameter space (using history matching). We showed that even when the prior assessment of model discrepancy is not inconsistent with the ability of the simulator, dimension reduction of spatial output using the ensemble-derived principal components (the SVD basis) often leads to a terminal case analysis.

We proposed a rotation of the SVD basis to highlight and incorporate important low-signal patterns that may be contained in the original SVD basis, giving a new basis that avoids the terminal case when this is possible. We then presented an efficient algorithm for optimal rotation, guaranteeing to avoid the terminal case when the model discrepancy allows, whilst ensuring enough signal on leading basis vectors to permit the fitting of emulators. We proved that our algorithm gives a valid rotation of the original basis, and established a fast test to see whether a given ensemble of model runs and discrepancy specification automatically leads to a terminal case analysis prior to rotation. Our methods are presented for models with spatial output, however, if basis methods were to be used for more general high dimensional output (e.g. spatio-temporal), the optimal rotation approach would not change if, for example, PCA were taken over the entire spatio-temporal output for the design, as in \cite{higdon2008computer}. 

We demonstrated the efficacy of our method using an idealised application, and showed that it scaled up to the important case of spatial output for state-of-the-art climate models. Our application highlighted the issue of the terminal case for climate model analyses, and showed the problems with using SVD in practice. We applied history matching for 2 waves to CanAM4 and showed how, combined with optimal rotation, we can begin to distinguish between what the modellers term `structural errors' and `tuning errors'. 

A purely methodological UQ approach for tuning climate models does not exist. It may be tempting, for UQ practitioners who are not familiar with climate models, to claim that calibration of computer simulators is a `solved' problem and that `all' that is required is for the modellers to specify their model discrepancy. We believe that the challenge for model tuning lies in the understanding of this elusive quantity. For the statistical community, rather than focussing on developing comprehensive methods for calibrating climate models automatically, this should mean engaging with modellers to develop robust tools and methods to help identify and understand these errors. This type of approach would have obvious implications for tuning, but would also feed into model development as it becomes better understood which parameters control various biases, and therefore which parameterisations need particular attention during the next development cycle.

\bibliographystyle{Chicago}

\bibliography{fullbib}

\pagebreak

\begin{center}
\textbf{\Large Supplemental Material: Uncertainty quantification for computer models with spatial output using calibration-optimal bases}
\end{center}
\setcounter{equation}{0}
\setcounter{figure}{0}
\setcounter{table}{0}
\setcounter{page}{1}
\setcounter{section}{0}
\makeatletter
\renewcommand{\thefigure}{S\arabic{figure}}
\renewcommand{\thesection}{S\arabic{section}}
\renewcommand{\theequation}{S\arabic{equation}}
\renewcommand{\theresult}{S\arabic{result}}

\section{Idealised example}

The spatial idealised example $f(\x)$, introduced in Section 3.1, gives output over a $10 \times 10$ field, and has 6 input parameters defined on $[-1, 1]$:
\begin{align}\label{toyfndef}
\begin{split}
f(\x) &= 3(10x_2^2 \bassf_2 + 5x_3^2 \bassf_2 + (x_3 + 1.5x_1 x_2) \bassf_3 + 2x_2 \bassf_4 + x_3 x_1 \bassf_5 +
(x_2 x_1) \bassf_6 + x_2^3 \bassf_7 \\
&+ (x_2 + x_3)^2 \bassf_8 + 2) + 1.5 \pi_N(x_4, 0.2, 0.1^2) \bassf_1 \frac{x_5}{1.3+x_6} + \Psi_{10 \times 10}(0, 0.05^2)
\end{split}
\end{align}
for $\pi_N(x_4, 0.2, 0.1^2)$ the density function of the Normal distribution with mean 0.2 and variance $0.1^2$, and where $\Psi_{10 \times 10}(0, 0.05^2)$ gives a sample from a Normal distribution with mean zero and variance $0.05^2$, at each location in the $10\times10$ grid. Figure \ref{toybasis} shows $(\bassf_1, \ldots, \bassf_8)$, with $\bassf_1$ giving the pattern most consistent with the observations, and $\bassf_2$ the basis vector that dominates the ensemble. After evaluating $f(\x)$, the output is vectorised so that $\ell = 100$.
\par
We define $\best$ as
\begin{displaymath}
\best = (0.7, 0.01, 0.01, 0.25, 0.8, -0.9) 
\end{displaymath}
with the observed field, $\obs$, given by adding a sample from $\mathrm{N}(\textbf{0}, \boldsymbol{\Sigma}_{\err})$ to $f(\best)$, to represent observation error. We define $\boldsymbol{\Sigma}_{\err}$ using the squared exponential correlation function over the $10\times10$ grid, with the spatial coordinates denoted by $\textbf{s}_i = (s_{i1}, s_{i2})$ for $i = 1, \ldots, 100$. The $i,j^{th}$ entry of $100\times 100$ matrix $\boldsymbol{\Sigma}_{\err}$ is therefore
\begin{equation} \label{sigmaerr}
\boldsymbol{\Sigma}_{\err}^{ij} = \text{exp} \{ -(s_{i1} - s_{j1})^2 - (s_{i2} - s_{j2})^2 \}.
\end{equation}
\begin{figure}[h]
\centering
\includegraphics[width = \linewidth]{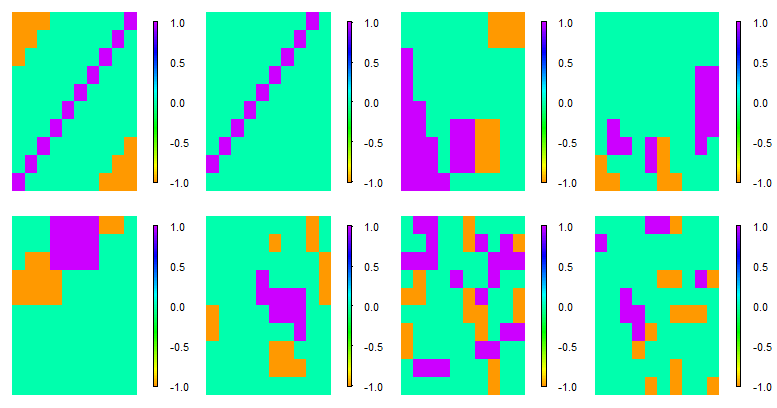}
\caption{The 8 orthogonal basis vectors used in the definition of the idealised function, with $\bassf_1$ shown in the top left, $\bassf_2$ to the right of this, etc.}
\label{toybasis}
\end{figure}
We model the discrepancy as
\begin{displaymath}
\disc \sim \mathrm{N}(\textbf{0}, \var_{\disc}),
\end{displaymath}
with the $(i, j)^{th}$ value of $\var_{\disc}$ given by
\begin{equation}
\var_{\disc}^{ij} = v_i v_j C(\textbf{s}_i , \textbf{s}_j )
\end{equation}
for variances $v_i, v_j$, and a correlation function $C(\cdot, \cdot)$ between locations $\textbf{s}_{i}$ and $\textbf{s}_{j}$. For $C(\cdot, \cdot)$, we again use the squared exponential correlation function, with the same correlation lengths as for $\var_{\err}$. We define $v_i$ via
\begin{displaymath}
v_{i} = \begin{cases}
0.1 \quad \text{if}\,\, i \in S \\
1 \quad \text{otherwise} \\
\end{cases}
\end{displaymath}
where the set $S$ contains the grid boxes on the main diagonal, as we are more interested in finding fields with output consistent with the observations in this region of the output. Figure \ref{truenroy} shows the true NROY space given $\var_{\err}$ and $\var_{\disc}$.
\begin{figure}[t]
\centering
\includegraphics[width = 0.75\textwidth]{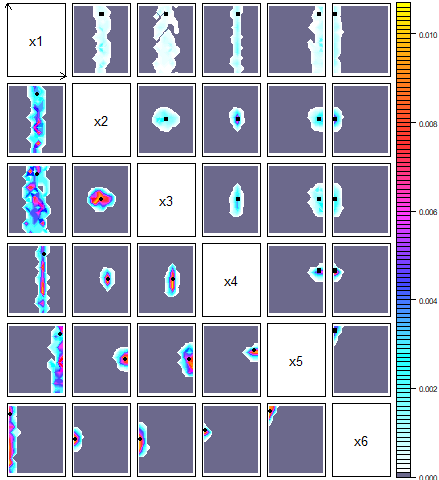}
\caption{Density plot of the true NROY space (i.e. no emulation involved), for each pair of parameters. For each cell in a particular pairwise plot, we average across the remaining 4 parameters, and plot the proportion of these runs that are in NROY space. The axes are reversed for the lower left plots, with the colour scale set individually for each plot. The black point corresponds to $\best$.}
\label{truenroy}
\end{figure}

\subsection{Probabilistic calibration for the SVD basis}

We performed probabilistic calibration following the Kennedy O'Hagan method, with fixed discrepancy as outlined in the previous section, uniform priors on the calibration parameters, and emulators as described in the main text.
\par
The dotted lines on Figure \ref{w1posteriors} show the resulting densities when  the truncated SVD basis, $\bas_4$, is used to calibrate probabilistically on the field, for the input parameters $x_1, \ldots, x_5$, and the ratio $r = x_5/(1.3+x_6)$. There are peaks of density away from the true values ($\best$, as shown by the red vertical lines), particularly for $x_3$ and $r$. For $x_4$, the parameter that controls the strength of the main diagonal, the posterior density is relatively flat across the entire range of $x_4$.
\begin{figure}[t]
\centering
\includegraphics[width = 0.75\textwidth]{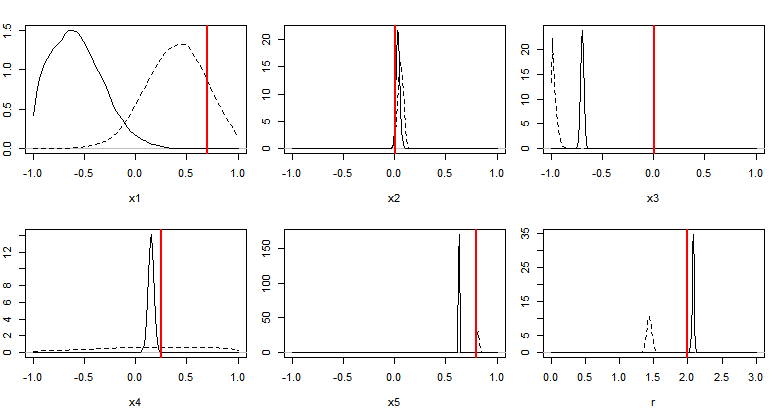}
\caption{The posterior distributions for $x_1, \ldots, x_5$ and $r$, when we calibrate probabilistically on the field with the SVD basis derived from the wave 1 ensemble, $\bas_4$ (dotted lines), and the wave 1 rotated basis (solid lines). The red vertical lines show the location of $\x^{*}$.}
\label{w1posteriors}
\end{figure}
\par
We sample from these posteriors and run the idealised function at these samples, to evaluate whether calibration with the truncated SVD basis has highlighted a region of parameter space that is `close' to $\obs$. 16 samples are shown in Figure \ref{svdcalibrationruns}, demonstrating that the results suggest it is not possible to remove the off-diagonal pattern that was dominant in the ensemble.
\begin{figure}[t]
\centering
\includegraphics[width = 0.7\textwidth]{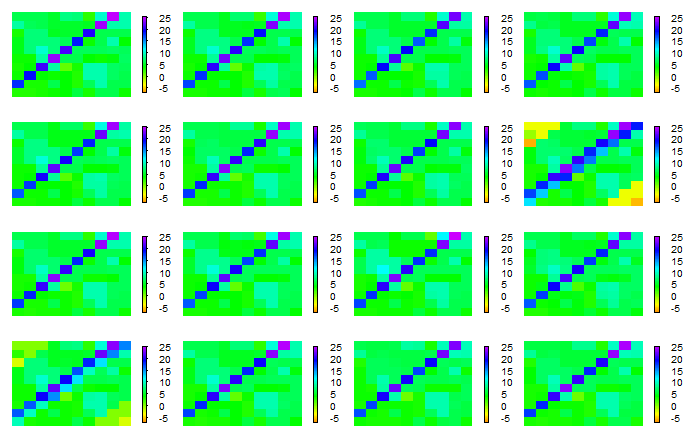}
\caption{$f(\x)$ at 16 samples of $\x$ from the calibration posterior distribution, when we emulate and calibrate with the truncated SVD basis $\bas_4$.}
\label{svdcalibrationruns}
\end{figure}

\subsection{Calibration with the rotated bases}

\begin{figure}[t]
\centering
\includegraphics[width = 0.75\textwidth]{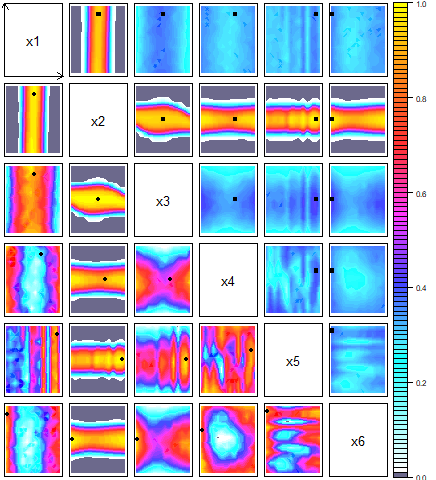}
\caption{Density plot for the wave 1 NROY space given by history matching using the rotated basis, for each pair of parameters. As in Figure S2, for each plot we average over the remaining parameters and plot the proportion in NROY space for each pair. The axes are reversed for the lower left plots, and the colour scales set differently. The black point corresponds to $\best$.}
\label{rotwave1nroy}
\end{figure}

The solid lines in Figure \ref{w1posteriors} show the posterior distributions for $x_1, \ldots, x_5$ and $r$ when the wave 1 rotated basis is used for probabilistic calibration, showing improvements (compared to the SVD basis) for $x_4$ and $r$.
\par
At wave 2, there are peaks of density at or near to the true parameter values for all but $x_5$ and $r$, as shown by the solid lines in Figure \ref{w23posteriors}. The dotted lines in this plot show the wave 3 posteriors. Although the peaks for $x_2$, $x_3$ and $x_4$ are not as large, this wave offers an improvement for $r$ (important for the strength of the main diagonal) and $x_1$ (as this parameter has no effect on the main or off-diagonal, a flat posterior is more accurate). The averaged posterior samples in Figure 5 show that the posterior distributions at each wave have identified better regions of parameter space than previously, with the wave 3 samples being consistent with the observations.

\begin{figure}[h!]
\centering
\includegraphics[width = 1\textwidth]{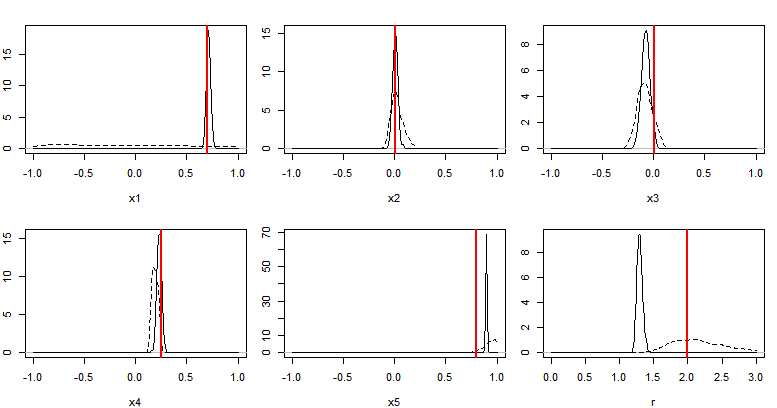}
\caption{The posterior distributions for $x_1, \ldots, x_5$ and $r$, at wave 2 (solid lines) and wave 3 (dotted lines), with the red lines equal to $\x^{*}$.}
\label{w23posteriors}
\end{figure}
\begin{figure}[t]
\centering
\includegraphics[width = 0.75\textwidth]{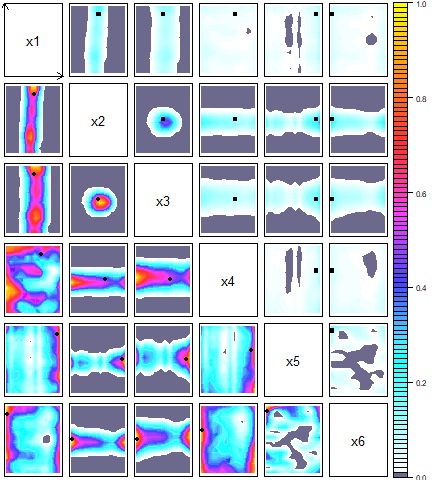}
\caption{Density plot for the wave 2 NROY space for each pair of parameters. Regions coloured grey indicate that there are no parameter settings in NROY space here, hence we see that we have significantly reduced space, compared to Figure S5.}
\label{rotwave2nroy}
\end{figure}
\begin{figure}[t]
\centering
\includegraphics[width = 0.75\textwidth]{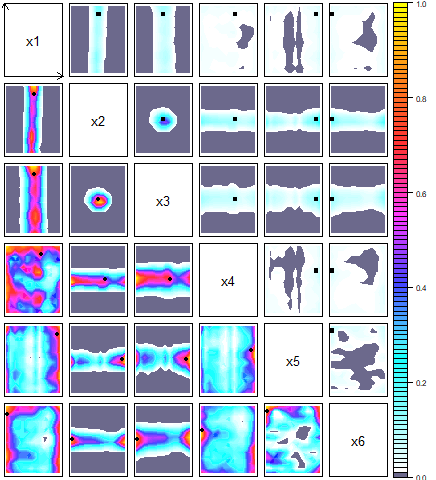}
\caption{Density plot for the wave 3 NROY space for each pair of parameters, with space again reduced from the previous wave (we have gone from 3.1\% to 2\% of the original parameter space).}
\label{rotwave3nroy}
\end{figure}
\par

%\subsection{Ensemble design} \label{ensdesign}

%After wave 1, subsequent ensembles are designed via the following procedure:
%\begin{enumerate}
%\item Let $\textbf{S}$ be a large sample of runs from $\mathcal{X}_{NROY}$.
%\item Define $m$ as the minimum implausibility given by points in $\textbf{S}$: \begin{displaymath}
%m = \min_{\x \in \textbf{S}} \mathcal{I}(\x)
%\end{displaymath}
%\item Divide \textbf{S} into $i$ groups of width $\omega = \frac{\ibound - m}{i}$ based on the value of $\mathcal{I}(\x)$, so that group $k$ is given by
%\begin{displaymath}
%s_k = \{\x \in \textbf{S} | m + (k - 1)\omega \leq \mathcal{I}(\x) < m + k\omega \}
%\end{displaymath}
%where $\ibound$ is the bound used to define $\mathcal{X}_{NROY}$.
%\item Sample $j$ points from each of the $i$ groups $s_k$, so that $ij = n$ is the ensemble size.
%\end{enumerate}
%As the size of $\textbf{S}$ tends to infinity, $\textbf{S} = \mathcal{X}_{NROY}$, so that $m$ is the minimum implausibility for any $\x$.
%\par
%Designing the ensemble in this way is similar to taking a Latin hypercube in implausibility space, rather than across the input space as is commonly done. For example, we could divide the implausibility space into $i = n$ groups, and, from each of these $n$ intervals, sample $j = 1$ value of $\x$, so that each implausibility interval is represented, as for Latin hypercubes.

\clearpage

\section{Proofs}

\subsection*{Weighted projection}

Define $\B = (\bb_1, \ldots, \bb_n)$, a basis with dimension $\ell \times n$, and $\weight$ an $\ell \times \ell$ positive definite matrix. The reconstruction $\rr = (r_1, \ldots, r_{\ell})^T$ of $\f = (f_1, \ldots, f_{\ell})^T$ has reconstruction error
\begin{equation} \label{supprecon}
\lVert \f - \rr \rVert_{\weight} = (\f - \rr)^T \weight^{-1} (\f - \rr)
\end{equation}
in the $\weight$ norm, where $\rr$ is given by coefficients $\cc = (c_1, \ldots, c_n)^T$ with
\begin{displaymath}
\rr = \sum_{k = 1}^n \bb_k c_k = \B \cc
\end{displaymath}
The vector $\cc$ that minimises the reconstruction error with respect to the $\weight$ norm is found by first writing \eqref{supprecon} in terms of $\cc$:
\begin{align*}
\begin{split}
(\f - \rr)^T \weight^{-1} (\f - \rr) &= (\f - \B \cc)^T \weight^{-1} (\f - \B \cc) \\
&= (\f^T - \cc^T \B^T) \weight^{-1} (\f - \B \cc) \\
&= \f^T \weight^{-1} \f - \cc^T \B^T \weight^{-1} \f - \f^T \weight^{-1} \B \cc + \cc^T \B^T \weight^{-1} \B \cc
\end{split}
\end{align*}
A scalar can be differentiated by a vector $\x$, with symmetric matrix $\textbf{A}$, via
\begin{align*}
\begin{split}
\frac{\partial}{\partial \x} \x^T \textbf{y} &= \frac{\partial}{\partial \x} \textbf{y}^T \x = \textbf{y}^T \\
\frac{\partial}{\partial \x} \x^T \textbf{A} \x &= 2 \x^T \textbf{A}
\end{split}
\end{align*}
Differentiating the reconstruction error with respect to $\cc$, with symmetric $\weight^{-1}$, we have
\begin{align*}
\begin{split}
\frac{\partial}{\partial \cc} \lVert \f - \rr \rVert_{\weight} &= 0 - (\B^T \weight^{-1} \f)^T - \f^T \weight^{-1} \B + 2 \cc^T \B^T \weight^{-1} \B \\
&= -\f^T \weight^{-1} \B - \f^T \weight^{-1} \B + \cc^T \B^T \weight^{-1} \B + \cc^T \B^T \weight^{-1} \B \\
&= -2\f^T \weight^{-1} \B + 2\cc^T \B^T \weight^{-1} \B
\end{split}
\end{align*}
Setting equal to zero, and solving for $\cc$, we have
\begin{align*}
\begin{split}
& 0 = -2\f^T \weight^{-1} \B + 2\hat{\cc}^T \B^T \weight^{-1} \B \\
&\implies \hat{\cc}^T \B^T \weight^{-1} \B = \f^T \weight^{-1} \B \\
&\implies \B^T \weight^{-1} \B \hat{\cc} = \B^T \weight^{-1} \f \\
&\implies \hat{\cc} = (\B^T \weight^{-1} \B)^{-1} \B^T \weight^{-1} \f
\end{split}
\end{align*}
With $\weight = \boldsymbol{\mathbb{I}}_{\ell}$, this is the usual projection equation:
\begin{displaymath}
\hat{\cc} = (\B^T \boldsymbol{\mathbb{I}}_{\ell}^{-1} \B)^{-1} \B^T \boldsymbol{\mathbb{I}}_{\ell}^{-1} \f = (\B^T \B)^{-1} \B^T \f
\end{displaymath}

\subsection*{Rotational invariance}

\textbf{Result 1} (Invariance of $\R_{\weight}(\cdot, \cdot)$ to rotation). \textit{For a rotation matrix $\rot$ of dimension $k \times k$, and a set of basis vectors $\B = (\bb_1, \ldots, \bb_n)$, we have}
	\begin{displaymath}
	\R_{\weight}(\B_k, \obs) = \R_{\weight}(\B_k \rot, \obs), \quad k = 1, \ldots, n.
	\end{displaymath}

\begin{proof}
We can rewrite the reconstruction error of the rotated basis, $\B_k \rot$, as
\begin{align*}
\begin{split}
\R_{\weight}(\B_k \rot, \obs) &= \lVert \obs - \B_k \rot ((\B_k \rot)^{T} \weight^{-1} \B_k \rot)^{-1} (\B_k \rot)^T \weight^{-1} \obs \rVert_{\weight} \\
&= \lVert \obs - \B_k \rot (\rot^T \B_k^T \weight^{-1} \B_k \rot)^{-1} \rot^T \B_k^T \weight^{-1} \obs \rVert_{\weight}
\end{split}
\end{align*}
Rotation matrix $\rot$ is invertible, with $\rot^T = \rot^{-1}$, by definition. We apply the identity $(\textbf{CD})^{-1} = \textbf{D}^{-1} \textbf{C}^{-1}$, where $\textbf{C}, \textbf{D}$ are $k \times k$ invertible matrices, for $\textbf{C} = \rot^T \B_k^T \weight^{-1} \B_k$ and $\textbf{D} = \rot$:
\begin{align*}
\begin{split}
&= \lVert \obs - \B_k \rot \rot^{-1} (\rot^T \B_k^T \weight^{-1} \B_k)^{-1} \rot^T \B_k^T \weight^{-1} \obs \rVert_{\weight} \\
&= \lVert \obs - \B_k (\B_k^T \weight^{-1} \B_k)^{-1} (\rot^T)^{-1} \rot^T \B_k^T \weight^{-1} \obs \rVert_{\weight} \\
&= \lVert \obs - \B_k (\B_k^T \weight^{-1} \B_k)^{-1} \B_k^T \weight^{-1} \obs \rVert_{\weight} \\
&= \R_{\weight}(\B_k, \obs)
\end{split}
\end{align*}
where in the second line, the inverse identity has been applied a second time with $\textbf{C} = \rot^T$ and $\textbf{D} = \B_k^T \weight^{-1} \B_k$, giving the final result.
\end{proof}

\subsection{Proof of Theorem 1}

Before proving Theorem 1, we first prove the following results:
\begin{result}[Orthogonality of the residual basis] \label{resultorth}
	The residual basis, $\B_{\epsilon}$, calculated from $\ensc$ and $\B_p$, is orthogonal to $\B_p$ (with respect to the $\weight$ norm).
\end{result}
\begin{proof}
First, we show the orthogonality (in $\weight$) of the columns of the residual ensemble, $\ens_{\resid}$, and the columns of $\B_p$:
\begin{align} \label{orthogresult}
\begin{split}
\B_p^T \weight^{-1} \ens_{\resid} &= \B_p^T \weight^{-1} (\ensc - \B_p (\B_p^{T} \weight^{-1} \B_p)^{-1} \B_p^{T} \weight^{-1} \ensc) \\
&= \B_p^T \weight^{-1} \ensc - (\B_p^T \weight^{-1} \B_p)(\B_p^{T} \weight^{-1} \B_p)^{-1} \B_p^{T} \weight^{-1} \ensc \\
&= \B_p^T \weight^{-1} \ensc - \B_p^{T} \weight^{-1} \ensc \\
&= \textbf{0}
\end{split}
\end{align}
This zero matrix has dimension $p \times n$, i.e. the basis vectors in $\B_p$ are orthogonal with the vectors of $\ens_{\resid}$, with respect to the $\weight$ norm. Using this, we obtain the result by considering the (generalised) singular value decomposition of $\ens_{\resid}^T$:
\begin{align*}
\begin{split}
\ens_{\resid}^T &= \textbf{U} \boldsymbol\Sigma \textbf{V}^{T} \\
\implies \ens_{\resid}^T &= \textbf{U} \boldsymbol\Sigma \B_{\resid}^{T}
\end{split}
\end{align*}
where $\textbf{U}$ is an orthonormal $n \times n$ matrix, $\boldsymbol\Sigma$ is a diagonal $n \times n$ matrix, and $\textbf{V} = \B_{\resid}$ is an $\ell \times n$ matrix with $\B_{\resid}^T \weight^{-1} \B_{\resid} = \boldsymbol{\mathbb{I}}_{n}$. From here, we have that
\begin{align*}
\begin{split}
\implies \ens_{\resid} &= \B_{\resid} \boldsymbol\Sigma^{T} \textbf{U}^{T} \\
\implies \ens_{\resid} \textbf{U} &= \B_{\resid} \boldsymbol\Sigma^{T} \textbf{U}^{T} \textbf{U} \\
\implies \B_p^T \weight^{-1} \ens_{\resid} \textbf{U} &= \B_p^T \weight^{-1} \B_{\resid} \boldsymbol\Sigma^{T} \\
\end{split}
\end{align*}
where we have multiplied on the left by $\B_p^T \weight^{-1}$. From \eqref{orthogresult}, we have $\B_p^T \weight^{-1} \ens_{\resid} = \textbf{0}$, hence 
\begin{displaymath}
\implies \B_p^T \weight^{-1} \B_{\resid} \boldsymbol\Sigma^{T} = \textbf{0}
\end{displaymath}
The final $p+1$ eigenvalues on the diagonal of $\boldsymbol\Sigma$ are zero (the $p$ vectors in $\B_p$, and the ensemble mean, remove $p+1$ degrees of freedom). Therefore, we are only interested in the leading $n - p - 1$ columns of $\B_{\resid}$, as all of the variability in $\ensc$ has already been explained. By discarding the columns associated with zero eigenvalues, we have that
\begin{align*}
\implies \B_p^T \weight^{-1} [\B_{\resid}]_{n-p-1} = \textbf{0}
\end{align*}
because $\boldsymbol\Sigma$ is diagonal, and hence we have the result.
\end{proof}
From Result \eqref{resultorth}, we can show that the basis vector selected at step $k$ of the optimal rotation algorithm, $\bass_k^* = \resbas^{k-1} \rott_k$, is orthogonal to those previously selected, $\bas^*_{k - 1} = (\bass_1^*, \ldots, \bass_{k - 1}^*)$:
\begin{displaymath}
(\bas^*_{k - 1})^T \weight^{-1} \bass_k^* = (\bas^*_{k - 1})^T \weight^{-1} \resbas^{k-1} \rott_k = \textbf{0}.
\end{displaymath}
\begin{result} \label{residbas1} When $\B$ is a basis for $\ensc$, we can write $\B_{\epsilon}^k = \B \rot_{\epsilon}^k$ for square $\rot_{\epsilon}^k$, i.e. the residual basis at iteration $k$ of the algorithm contains linear combinations of the vectors of $\B$, and hence each vector selected by the algorithm is a linear combination of $\B$.
\end{result}
\begin{proof}
	By the singular value decomposition of the residual ensemble after selecting $k$ basis vectors, we have
	\begin{displaymath}
	(\resens^k)^T = \textbf{U}_{\epsilon}^k \boldsymbol{\Sigma}_{\epsilon}^k (\resbas^k)^T 
	\end{displaymath}
	for orthonormal $\textbf{U}_{\epsilon}^k$, and diagonal $\boldsymbol{\Sigma}_{\epsilon}^k$. We can write $\ensc = \B \rot_{\boldsymbol{\mu}}$ ($\B$ is a basis for $\ensc$, hence the ensemble is a linear combination of the basis vectors), for $n \times n$ matrix $\rot_{\boldsymbol{\mu}}$. At iteration $k$ of the optimal rotation algorithm, we have
	\begin{align*}
	\begin{split}
	\resbas^k \boldsymbol{\Sigma}_{\epsilon}^k (\textbf{U}_{\epsilon}^k)^T &= \ensc - \bas^*_k ((\bas^*_k)^T \weight^{-1} \bas^*_k)^{-1} (\bas^*_k)^{T} \weight^{-1} \ensc \\
	\implies \resbas^k &= (\B \rot_{\boldsymbol{\mu}} - \bas^*_k ((\bas^*_k)^T \weight^{-1} \bas^*_k)^{-1} (\bas^*_k)^{T} \weight^{-1} \B \rot_{\boldsymbol{\mu}}) \textbf{U}_{\epsilon}^k (\boldsymbol{\Sigma}_{\epsilon}^k)^{-1} \\
	\end{split}
	\end{align*} 
	Set $k = 1$, i.e. we have only selected one basis vector so far. When $k = 1$, the algorithm selects a linear combination of $\B$ by construction, hence we have $\bass^*_1 = \bas^*_1 = \B \tilde{\rott}_1$ for vector $\tilde{\rott}_1$. Therefore,
	\begin{align} \label{resbasproof}
	\begin{split}
	\resbas^1 &= (\B \rot_{\boldsymbol{\mu}} - \B \tilde{\rott}_1 ((\B \tilde{\rott}_1)^T \weight^{-1} \B \tilde{\rott}_1)^{-1} (\B \tilde{\rott}_1)^{T} \weight^{-1} \B \rot_{\boldsymbol{\mu}}) \textbf{U}_{\epsilon}^1 (\boldsymbol{\Sigma}_{\epsilon}^1)^{-1} \\
	&= \B (\rot_{\boldsymbol{\mu}} - \tilde{\rott}_1 ((\B \tilde{\rott}_1)^T \weight^{-1} \B \tilde{\rott}_1)^{-1} (\B \tilde{\rott}_1)^{T} \weight^{-1} \B \rot_{\boldsymbol{\mu}})  \textbf{U}_{\epsilon}^1 (\boldsymbol{\Sigma}_{\epsilon}^1)^{-1} \\
	&= \B \rot_{\epsilon}^1
	\end{split}
	\end{align}
	with $\rot_{\epsilon}^1 = \rot_{\boldsymbol{\mu}} - \tilde{\rott}_1 ((\B \tilde{\rott}_1)^T \weight^{-1} \B \tilde{\rott}_1)^{-1} (\B \tilde{\rott}_1)^{T} \weight^{-1} \B \rot_{\boldsymbol{\mu}})  \textbf{U}_{\epsilon}^1 (\boldsymbol{\Sigma}_{\epsilon}^1)^{-1}$.
	Then at iteration $k = 2$, we optimise over linear combinations of $\resbas^1$, so that
	\begin{displaymath}
	\bass_2^* = \resbas^1 \rott_2 = \B \rot_{\epsilon}^1 \rott_2 = \B \tilde{\rott_2}
	\end{displaymath}
	i.e. the second basis vector is a linear combination of $\B$. It follows that at iteration $k$, we select a new basis vector where
	\begin{displaymath}
	\bass_k^* = \resbas^{k-1} \rott_k = \B \rot_{\epsilon}^{k-1} \rott_k = \B \tilde{\rott_k}
	\end{displaymath}
	for $\rot_{\epsilon}^{k-1} = \rot_{\boldsymbol{\mu}} - \tilde{\rot}_{k-1} ((\B \tilde{\rot}_{k-1})^T \weight^{-1} \B \tilde{\rot}_{k-1})^{-1} (\B \tilde{\rot}_{k-1})^{T} \weight^{-1} \B \rot_{\boldsymbol{\mu}})  \textbf{U}_{\epsilon}^1 (\boldsymbol{\Sigma}_{\epsilon}^1)^{-1}$, and $\tilde{\rot}_{k-1} = (\tilde{\rott}_1, \ldots, \tilde{\rott}_{k-1})$ (so that $\bas^*_{k-1} = \B \tilde{\rot}_{k-1}$). Therefore, we have that the residual basis at each iteration is a linear combination of the original basis, $\B$, and hence each new basis vector is.
\end{proof}

\noindent
\textbf{Theorem 1.} \textit{$\bas^*$ in step 3 of the optimal rotation algorithm is an orthogonal rotation of $\B$.}
\begin{proof}
	Assume that we have performed $k$ iterations of the algorithm, resulting in the basis
	\begin{equation} \label{newbasis}
	\bas^* = (\bass_1^*, \ldots, \bass_k^*, [\B_{\epsilon}^k]_{n-k})
	\end{equation}
	With $\B_{\epsilon}^k = \B \rot_{\epsilon}^k$ and $\bass_j^* = \B \tilde{\rott_j}$ (Result \ref{residbas1}), we rewrite \eqref{newbasis} as
	\begin{displaymath}
	\bas^* = (\B \tilde{\rott}_1, \ldots, \B \tilde{\rott}_k, [\B \rot_{\epsilon}^k]_{n-k})
	= \B (\tilde{\rott}_1, \ldots, \tilde{\rott}_k, [\rot_{\epsilon}^k]_{n-k}) = \B\rot.
	\end{displaymath}
	We show that $\rot^T \rot = \boldsymbol{\mathbb{I}}_n$, i.e. $\rot$ is a rotation matrix. We have
	\begin{equation} \label{rotproof}
	\rot^T \rot = \begin{pmatrix}
	\tilde{\rott}_1^T \tilde{\rott}_1 & \tilde{\rott}_1^T \tilde{\rott}_2 & \ldots & \tilde{\rott}_1^T [\rot_{\epsilon}^k]_{n-k} \\
	\tilde{\rott}_2^T \tilde{\rott}_1 & \tilde{\rott}_2^T \tilde{\rott}_2 & \ldots & \tilde{\rott}_2^T [\rot_{\epsilon}^k]_{n-k} \\
	\vdots & & \ddots & \\
	[\rot_{\epsilon}^k]_{n-k}^T \tilde{\rott}_1 & [\rot_{\epsilon}^k]_{n-k}^T \tilde{\rott}_2 & \ldots & [\rot_{\epsilon}^k]_{n-k}^T [\rot_{\epsilon}^k]_{n-k} \\
	\end{pmatrix}
	\end{equation}
	The upper-left $k \times k$ block can be written as 
	\begin{displaymath}
	\tilde{\rott}_i^T \tilde{\rott}_j = \tilde{\rott}_i^T \B^T \weight^{-1} \B \tilde{\rott}_j = (\bass_i^*)^T \weight^{-1} \bass_j^* = \begin{cases}
	1 \quad \text{if} \quad i = j \\
	0 \quad \quad \text{otherwise.}
	\end{cases}
	\end{displaymath}
	Similarly,
	\begin{displaymath}
	[\rot_{\epsilon}^k]_{n-k}^T \tilde{\rott}_j = (\B [\rot_{\epsilon}^k]_{n-k})^T \weight^{-1} \bass_j^* = \B_{\epsilon}^T \weight^{-1} \bass_j^* = \textbf{0},
	\end{displaymath}
	by Result \ref{resultorth}. Finally,
	\begin{displaymath}
	[\rot_{\epsilon}^k]_{n-k}^T [\rot_{\epsilon}^k]_{n-k} = [\rot_{\epsilon}^k]_{n-k}^T \B^T \weight^{-1} \B [\rot_{\epsilon}^k]_{n-k} = [\B_{\epsilon}^k]_{n - k}^T \weight^{-1} [\B_{\epsilon}^k]_{n - k} = \boldsymbol{\mathbb{I}}_{n - k}
	\end{displaymath}
	and hence from \eqref{rotproof} we have $\rot^T \rot = \boldsymbol{\mathbb{I}}_{n}$, and $\rot$ is a rotation matrix.
\end{proof}

\subsection{Gram-Schmidt invariance}

Gram-Schmidt orthonormalisation imposes orthonormality on basis vectors $\B = (\bb_1, \ldots, \bb_n)$ \citep{bjorck1967solving}. It can be written in terms of matrices \citep{bjorck1994numerics}:
\begin{displaymath}
\B = \bas \textbf{R}
\end{displaymath}
where $\bas$ is an $l \times n$ basis containing normalised, orthogonal vectors $\bass_1, \ldots, \bass_n$, and $\textbf{R}$ is an $n \times n$ upper-triangular matrix. Therefore, the $j^{th}$ new basis vector is a linear combination of the first $j$ basis vectors of $\B$.
\begin{result}[Gram-Schmidt invariance] \label{resultgs}
The reconstruction given by the first $q$ vectors of the original basis is equal to the reconstruction given by the first $q$ vectors of the orthogonal basis:
\begin{displaymath}
\R_{\weight}(\B_q, \obs) = \R_{\weight}(\bas_q, \obs), \quad q = 1, \ldots, n
\end{displaymath}
\end{result}
\begin{proof}
Using $\bas^T \weight^{-1} \bas = \boldsymbol{\mathbb{I}}_{n}$ (i.e. we've imposed orthonormality with respect to the $\weight$ norm), we have
\begin{align*} 
\begin{split}
\B (\B^T \weight^{-1} \B)^{-1} \B^T \weight^{-1} \obs &= \bas \textbf{R} ((\bas \textbf{R})^T \weight^{-1} \bas \textbf{R})^{-1} (\bas \textbf{R})^T \weight^{-1} \obs \\
&= \bas \textbf{R} (\textbf{R}^T \bas^T \weight^{-1} \bas \textbf{R})^{-1} \textbf{R}^T \bas^T \weight^{-1} \obs \\
&= \bas \textbf{R} (\textbf{R}^T\textbf{R})^{-1} \textbf{R}^T \bas^T \weight^{-1} \obs \\
\end{split}
\end{align*}
As $\textbf{R}$ is invertible, and applying the identity $(\textbf{C} \textbf{D})^{-1} = \textbf{D}^{-1} \textbf{C}^{-1}$ for square matrices $\textbf{C}$ and $\textbf{D}$, we have
\begin{align*}
\begin{split}
&=\bas \textbf{R} \textbf{R}^{-1} \textbf{R}^{-T}  \textbf{R}^T \bas^T \weight^{-1} \obs \\
&= \bas \bas^T \weight^{-1} \obs \\
&= \bas (\bas^T \weight^{-1} \bas)^{-1} \bas^T \obs
\end{split}
\end{align*}
i.e. the reconstruction of $\obs$ is the same with both $\B$ and $\bas$. The proof proceeds analogously for any truncation of these bases.
\end{proof}
Alternatively, this result could be shown by proving that both $\B$ and $\bas$ span the same subspace, and using that a basis gives unique representations of general fields \citep{kuttler2012linear}.

%Given the wave 1 NROY space, the wave 2 ensemble is designed with a similar approach as for the idealised example (Section \ref{ensdesign}), with runs selected to ensure that both a space-filling design in NROY space across the input parameters, and across the implausibilities for each of the three fields, is achieved. Furthermore, as there are only 50 runs, the expected fields are checked, with parameter settings expected to fix biases in the ensemble selected, with care taken to ensure that different biases should be fixed across the new ensemble. The input parameters are shown in Figure \ref{w2design}.
%\begin{figure}[h!]
%\centering
%\includegraphics[width = 0.9\linewidth]{w2design}
%\caption{Pairs plot showing the wave 2 design for CanAM4.}
%\label{w2design}
%\end{figure}

\clearpage
\section{CanAM4}

\begin{figure}[h]
\centering
\includegraphics[width = 0.7\linewidth]{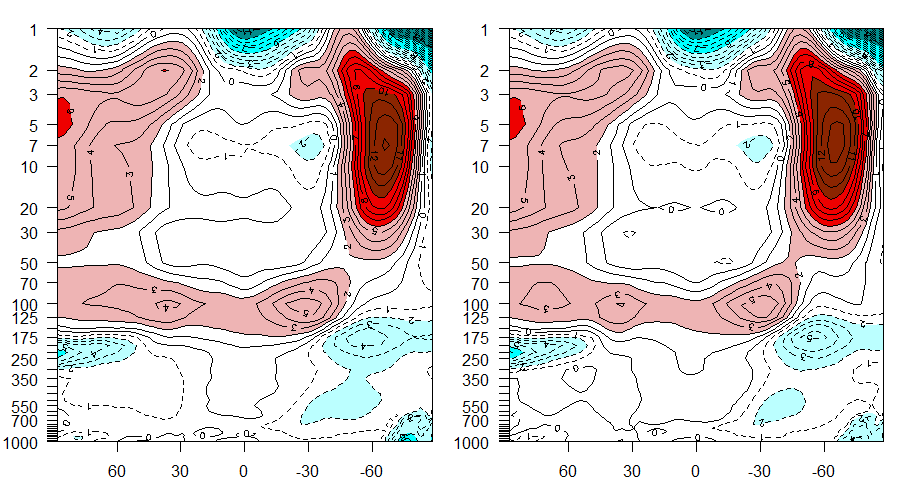}
\caption{The TA anomaly for the standard run (left), and for run 005 of the new ensemble, the `best' run in the wave 2 ensemble, in terms of minimising the root mean squared error. The large warm bias from the standard run has not been removed.}
\label{TAbest}
\end{figure}

\begin{figure}[h!]
	\centering
	\includegraphics[width = 0.4\linewidth]{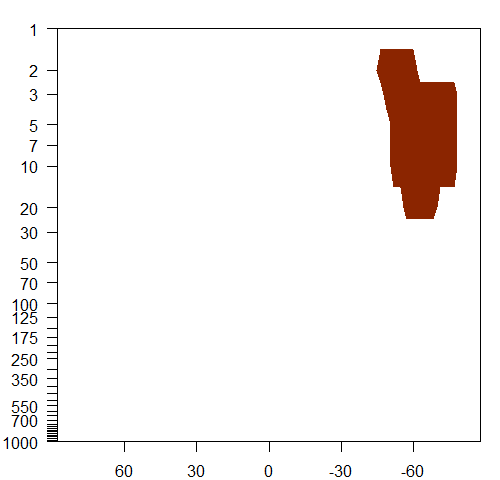}
	\caption{The grid boxes where we deem there to be a potential structural error for TA, and hence where we increase the discrepancy, as described in Section 5.}
	\label{tadiscrepancy}
\end{figure}

\end{document}